%% file: reduction_arXiv.tex
\documentclass[twocolumn,showpacs,notitlepage,nofootinbib,floatfix,superscriptaddress]{revtex4-2} 
\usepackage{graphicx, xcolor}
\usepackage{amssymb, amsmath, amsthm, amsfonts}
\usepackage{hyperref} 
\usepackage{tikz}

\newtheorem{theorem}{Theorem}
\theoremstyle{theorem}
\newtheorem{proposition}[theorem]{Proposition}
\theoremstyle{proposition}
\newtheorem{lemma}[theorem]{Lemma}
\theoremstyle{lemma}

\newcommand{\T}{\mathrm{T}}
\newcommand{\zero}{0}
\newcommand{\one}{I}
\newcommand{\dis}[2][]{\Delta_{#1}\!\left(#2\right)}
\DeclareMathOperator*{\supp}{supp}
\DeclareMathOperator*{\lcm}{lcm}
\DeclareMathOperator*{\spn}{span}
\DeclareMathOperator*{\poly}{poly}
\DeclareMathOperator*{\argmin}{arg\,min}

\begin{document}

\title{Finding the disjointness of stabilizer codes is NP-complete}
	
\author{John Bostanci}
\affiliation{Institute for Quantum Computing, University of Waterloo, Waterloo, ON N2L 3G1, Canada}
\author{Aleksander Kubica}
\affiliation{Perimeter Institute for Theoretical Physics, Waterloo, ON N2L 2Y5, Canada}
\affiliation{Institute for Quantum Computing, University of Waterloo, Waterloo, ON N2L 3G1, Canada}
\affiliation{AWS Center for Quantum Computing, Pasadena, CA 91125, USA}
\affiliation{California Institute of Technology, Pasadena, CA 91125, USA}

\begin{abstract}
The disjointness of a stabilizer code is a quantity used to constrain the level of the logical Clifford hierarchy attainable by transversal gates and constant-depth quantum circuits.
We show that for any positive integer constant $c$, the problem of calculating the $c$-disjointness, or even approximating it to within a constant multiplicative factor, is NP-complete.
We provide bounds on the disjointness for various code families, including the CSS codes, concatenated codes and hypergraph product codes.
We also describe numerical methods of finding the disjointness, which can be readily used to rule out the existence of any transversal gate implementing some non-Clifford logical operation in small stabilizer codes.
Our results indicate that finding fault-tolerant logical gates for generic quantum error-correcting codes is a computationally challenging task.
\end{abstract}

\maketitle


Designing fault-tolerant schemes is an essential step toward scalable universal quantum computation~\cite{Shor1996,Steane1997,Knill2005,Knill2005a}.
To protect quantum information from the detrimental effects of noise one typically encodes it into an quantum error-correcting code.
In addition to reliably storing quantum information, one also seeks to perform fault-tolerant logical operations on the encoded information.

One of the simplest ways to realize fault-tolerant logical operations is via transversal gates, which act independently on individual physical qubits and thus do not spread errors in an uncontrollable way.
Recently, many works have been devoted to transversal gates implementing non-Clifford logical operations in topological quantum codes~\cite{Bombin2013,Kubica2015a,Watson2015,Kubica2015,Bombin2018,Jochym-OConnor2021,Vasmer2021}
and the consequent universal quantum computation schemes
~\cite{Bombin2016dim,Bombin2018a,Kubicathesis,Vasmer2019,Brown2020,Iverson2021}.
Transversal gates also prove useful for magic state distillation~\cite{Knill2004,Knill2004a,Bravyi2005}, as they form the backbone of many distillation protocols; see Ref.~\cite{Beverland2021} and the references therein.

Logical operations implemented via transversal gates are somewhat limited.
Namely, the computational universality of transversal gates is incommensurate with the capability of the underlying code to correct errors, as exemplified by the Eastin-Knill theorem~\cite{Eastin2009,Zeng2011}, and its approximate versions~\cite{Faist2019,Woods2020,Kubica2020}.
More generally, bounded-spread logical operators, which propagate errors in a benign way and include constant-depth quantum circuits and locality-preserving operators, are also computationally limited~\cite{Bravyi2013,Pastawski2014,Beverland2014,Jochym-OConnor2018,Webster2020}.

Although systematic approaches to finding transversal logical gates for generic quantum error-correcting codes are not known, for stabilizer codes~\cite{Gottesman1996} we can rule out the possibility of implementing certain
logical operations.
Namely, if $M$ is the level of the logical Clifford hierarchy~\cite{Gottesman1999a} attainable by transversal logical gates, then the following upper bound holds
\begin{equation}
\label{eq_level_bound}
M \leq \left\lfloor \log_\Delta \left(d_\uparrow/d_\downarrow\right)\right\rfloor + 2,
\end{equation}
given the min-distance $d_\downarrow > 1$, max-distance $d_\uparrow$ and disjointness $\Delta$ of the stabilizer code \cite{Jochym-OConnor2018}.
The disjointness, roughly speaking, captures the maximal number of mostly non-overlapping representatives of any given non-trivial logical Pauli operator.
Until now, however, the problem of finding the disjointness as well as its computational hardness have not been explored.

In our work we focus on the problem of finding the disjointness of stabilizer codes, which serves as a proxy to understanding what are the admissible fault-tolerant logical gates.
First, in Section~\ref{sec_hardness} we show that for any positive integer $c$ it is NP-complete to calculate the $c$-disjointness, as well as to approximate it to within any constant multiplicative factor.
Our result thus indicates that finding fault-tolerant logical gates that can be implemented with generic quantum error-correcting codes is a computationally challenging task.
Then, in Section~\ref{sec_practice} we discuss numerical methods of finding the disjointness, which we illustrate with the example of the $[\![ 14,3,3 ]\!]$ stabilizer code~\cite{Landahl2020}.
We also provide a strengthening of the bound in Eq.~\eqref{eq_level_bound}, which subsequently rules out the existence of any transversal logical non-Clifford gate in the aforementioned $[\![ 14,3,3 ]\!]$ stabilizer code.
Lastly, in Section~\ref{sec_bounds} we provide bounds on the disjointness for various code families, including the CSS codes~\cite{Calderbank1996, Steane1996CSS}, concatenated codes~\cite{Knill1996} and hypergraph product codes~\cite{Tillich2014}.

\section{Preliminaries}
\label{section:preliminaries}

In this section, we briefly discuss basic constructions of stabilizer codes, as well as the notions of code distance and disjointness.
We also comment on certain graph-theory problems and their computational complexity.

\subsection{Stabilizer code constructions}

Stabilizer codes are an important class of quantum error-correcting codes.
A stabilizer code is defined by its stabilizer group $\mathcal S$, i.e., an Abelian subgroup of the Pauli group that does not contain $-I$.
In what follows we identify the stabilizer code with its stabilizer group.
The code space of the stabilizer code $\mathcal S$ is the simultaneous $(+1)$-eigenspace of all of the stabilizer operators.
We denote by $[\![n, k, d]\!]$ a stabilizer code that encodes $k$ logical qubits into $n$ physical qubits and has code distance $d$.
To specify the $[\![ n,k,d]\!]$ stabilizer code, we can provide a binary matrix of size $(n-k)\times 2n$, whose rows correspond to independent stabilizer generators of $\mathcal S$.
For concreteness, we identify a bit string $(b_1,\ldots,b_{2n})\in \{0,1\}^{2n}$ with the following Pauli operator $\bigotimes_{i=1}^n X^{b_i}_i Z^{b_{i+n}}_i$, where $P_i$ denotes a Pauli $P\in\{X,Z\}$ operator acting on qubit $i\in[n] = \{1,\ldots,n\}$.
We say that an operator is of $X$- or $Z$-type iff it is a tensor product of either Pauli $X$ or Pauli $Z$ operators (and the identity operators).

For any stabilizer code, logical Pauli operators, which are the elements of the the normalizer of the stabilizer group in the Pauli group, can always be implemented as tensor products of single-qubit Pauli operators.   
We write $\overline L$ to represent a logical Pauli operator itself, as well as the set of its equivalent representatives.
Also, we write $\mathcal L$ to denote the set of all non-trivial logical Pauli operators.

{\bf CSS codes.---}A stabilizer code is a CSS code iff there exists a choice of stabilizer generators such that every generator is either a Pauli $X$- or $Z$-type operator.
Given a CSS code with code parameters $[\![n,k,d]\!]$, we can always choose its logical Pauli operators in a way that for every $i\in [k]$ the logical Pauli $\overline X_i$ and $\overline Z_i$ operators are implemented via Pauli $X$- and $Z$-type operators, respectively.
We refer to such a set of logical operators $\{\overline X_i, \overline Z_i \}_{i\in [k]}$ as a standard logical basis.
Then, for every logical Pauli operator $\overline L$ we can find its decomposition in a standard logical basis, i.e., $\overline L = \overline L^X \overline L^Z$, where $\overline L^X$ and $\overline L^Z$ are $X$- and $Z$-type logical Pauli operators, respectively.

{\bf Concatenated stabilizer codes.---}Given two stabilizer codes $\mathcal S_1$ and $\mathcal S_2$ with parameters $[\![n_1, k_1, d_1]\!]$ and $[\![n_2, 1, d_2]\!]$, respectively, we can concatenate them to obtain a new stabilizer code $\mathcal S_1 \lhd \mathcal S_2$.  
To construct the concatenated code $\mathcal S_1 \lhd \mathcal S_2$, we first encode the logical information into the stabilizer code $\mathcal S_1$, then we encode each and every qubit of the stabilizer code $\mathcal S_1$ into the stabilizer code $\mathcal S_2$.
The concatenated code $\mathcal S_1 \lhd \mathcal S_2$ is a stabilizer code with parameters $[\![n_1 n_2, k_1, d]\!]$, where $d \geq d_1 d_2$.

{\bf Hypergraph product codes.---}Given two full-rank binary matrices $H_1$ and $H_2$ of size $m_1\times n_1$ and $m_2\times n_2$, respectively, the corresponding hypergraph product code is specified by the following binary matrix
\begin{equation}
\label{eq_hypergraph_matrix}
\left(\begin{array}{c|c|c|c}
H_1 \otimes \one_{m_2} & \one_{m_1} \otimes H_2
& \multicolumn{2}{c}{\zero_{m_1m_2, n_1m_2+n_2 m_1}}\\
\hline
\multicolumn{2}{c|}{\zero_{n_1n_2, n_1m_2+n_2 m_1}}
&\one_{n_1}\otimes H_2^{\T} & H_1^{\T} \otimes I_{n_2}
\end{array}\right),
\end{equation}
where $H_*^\T$ denotes the transpose of $H_*$, $\zero_{a, b}$ and $\one_a$ are the zero matrix and the identity matrix of size $a\times b$ and $a\times a$, respectively.
Note that hypergraph product codes are CSS codes.

\subsection{Distance and disjointness}

Let $\overline L \in \mathcal L$ be any non-trivial logical Pauli operator for the stabilizer code $\mathcal S$.
Following Ref.~\cite{Jochym-OConnor2018}, we define the distance $d(\overline L)$ to be the size of the support of the smallest representative of $\overline L$, i.e.,
\begin{equation}
d(\overline L) = \min_{L\in \overline L} |\supp L|,
\end{equation}
and introduce the notions of the min-distance and max-distance as follows
\begin{equation}
d_\downarrow = \min_{\overline L \in \mathcal L} d(\overline L),\quad d_\uparrow = \max_{\overline L \in \mathcal L} d(\overline L).
\end{equation}
Note that the min-distance is the same as the standard stabilizer code distance.

Let $\mathcal A \subseteq \overline L$ be a subset of representatives of $\overline L$ or, more generally, a multiset that allows for multiple instances for each of the representatives of $\overline L$.
We say that $\mathcal A$ is a $c$-disjoint collection of representatives of $\overline L$, where $c$ is a positive integer, iff for every qubit there are at most $c$ elements of $\mathcal A$ that are supported on that qubit.
We then define the $c$-disjointness $\dis[c]{\overline L}$ to be the size of the largest $c$-disjoint collection for $\overline{L}$ divided by $c$, i.e.,
\begin{eqnarray}
\Delta_c(\overline{L}) = c^{-1}\max_{\mathcal A \subseteq \overline{L}}\{|\mathcal A|
&\mathrel{:}& \text{ at most $c$ elements $L \in \mathcal A$ }\\
&& \text{are supported on any qubit}\}.\quad
\end{eqnarray}
Subsequently, the disjointness $\dis{\mathcal S}$ of the stabilizer code $\mathcal S$ is defined as follows
\begin{equation}
\label{eq_disjointness_new}
\Delta(\mathcal S) = \sup_{c \geq 1} \min_{\overline{L} \in \mathcal L} \Delta_c(\overline{L}).
\end{equation}
Note that the disjointness defined here, which allows representatives of $\overline L$ to be selected multiple times, is greater or equal to the disjointness defined in Ref.~\cite{Jochym-OConnor2018}.
Thus, by using our definition in Eq.~\eqref{eq_level_bound} we may obtain a tighter bound on the level of the logical Clifford hierarchy attainable by transversal logical gates.
At the same time, it is not obvious that the supremum in Eq.~\eqref{eq_disjointness_new} can be admitted for some finite positive integer $c$.
We establish this fact in Theorem~\ref{theorem_finite_c} and Proposition~\ref{proposition_disjointness_equiv}.

Finally, for any positive integer $c$ we introduce the following decision problem based on the $c$-disjointness.
\begin{quote}
\textbf{$c$-DISJOINTNESS}\\
\textbf{Input}: A full-rank binary matrix of size $(n-k) \times 2n$ specifying a stabilizer code $\mathcal S$,
a string of $2n$ bits representing a logical Pauli operator $\overline{L}$ and a positive integer $a$.\\
\textbf{Question}: Is the size of the largest $c$-disjoint collection of representatives of $\overline L$ greater or equal to $a$, i.e.,
$c\dis[c]{\overline{L}} \geq a$?
\end{quote}

\begin{figure*}[ht!]
\centering
(a)\includegraphics[width=0.27 \textwidth]{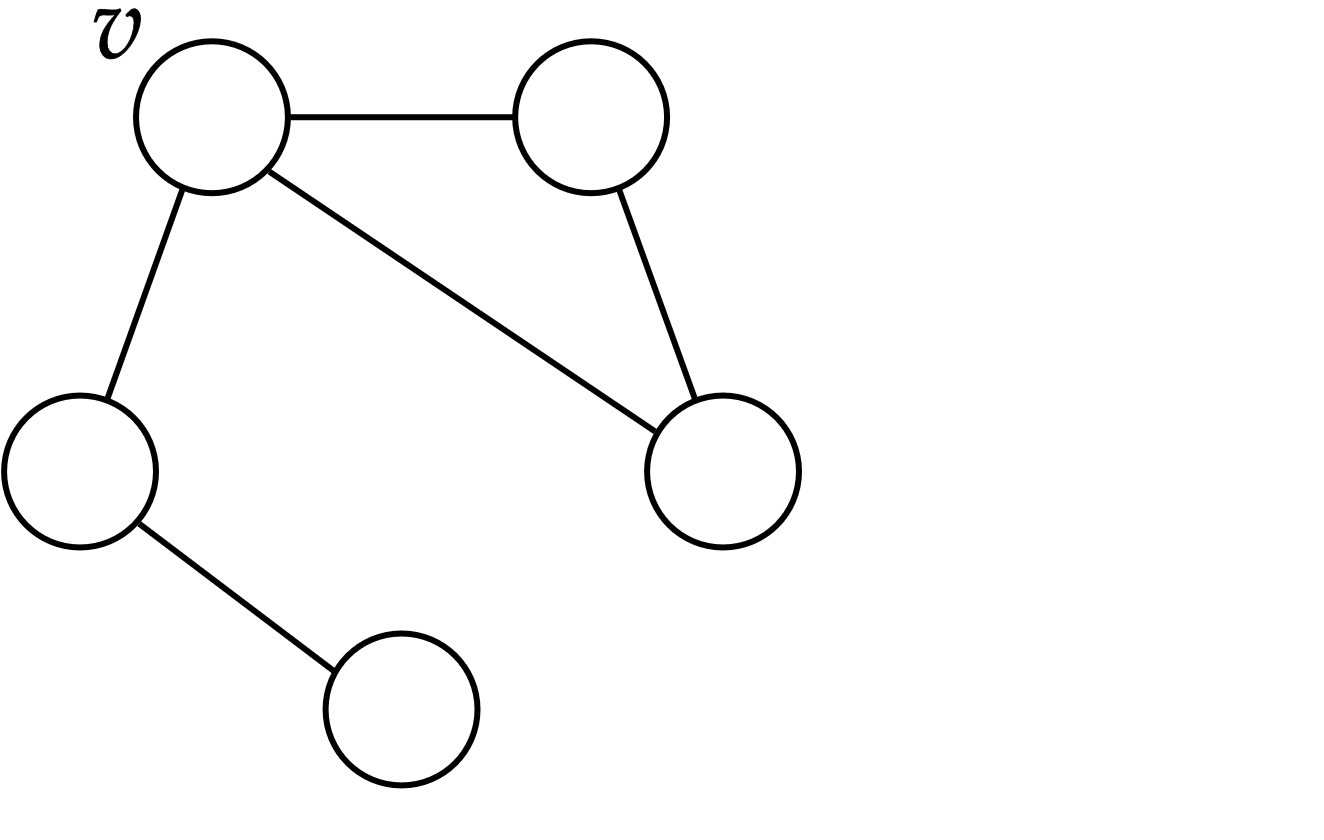}
(b)\includegraphics[width=0.27 \textwidth]{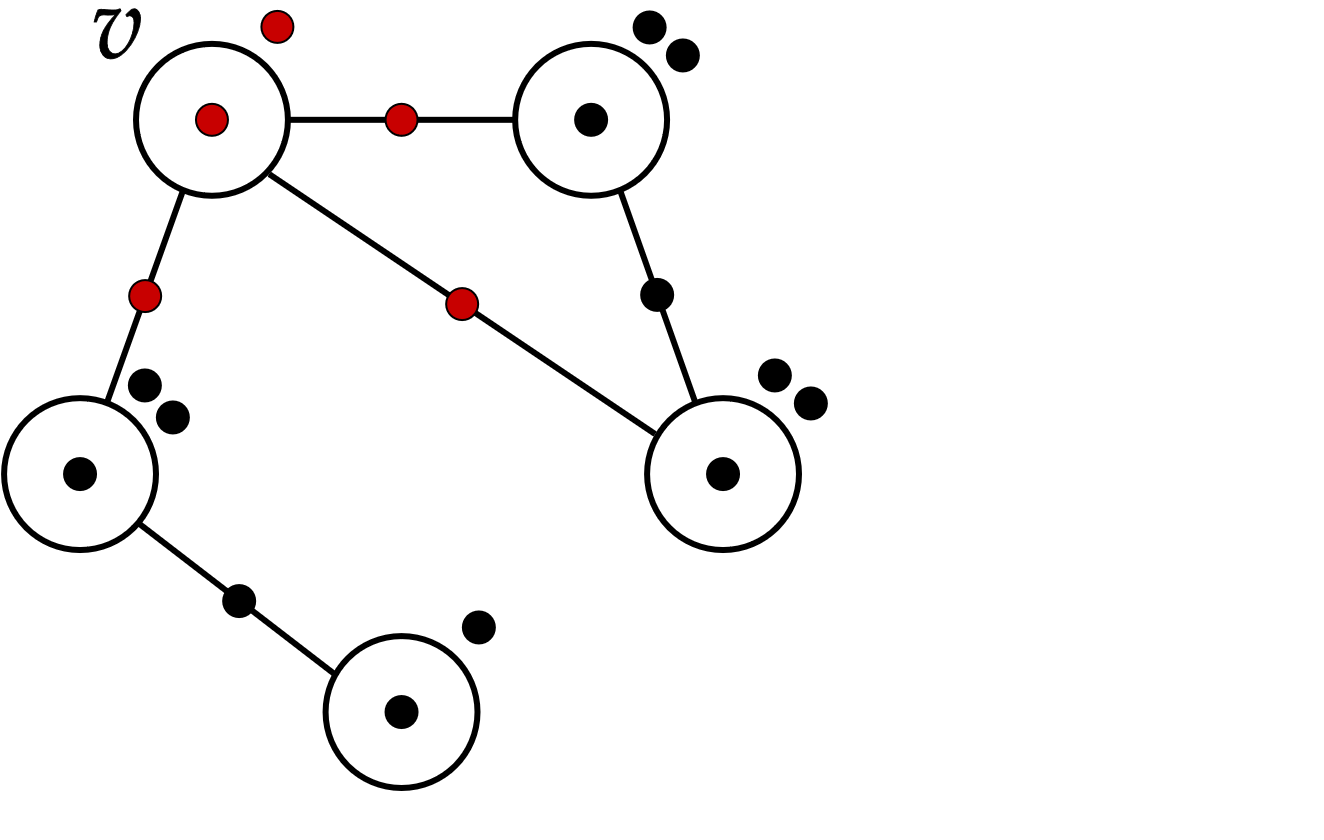}
(c)\includegraphics[width=0.27 \textwidth]{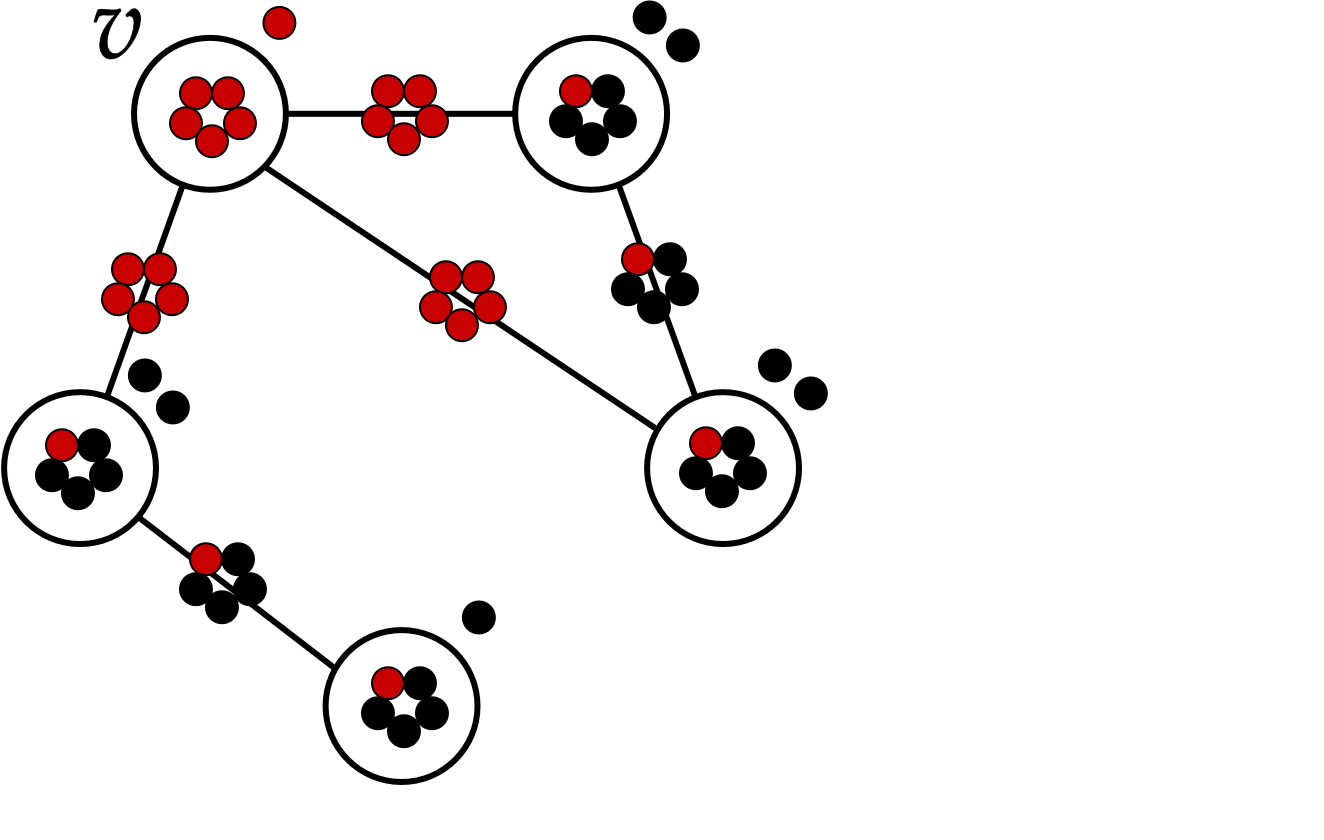}
\caption{
(a) A graph $G = (V,E)$ can be used to define a CSS stabilizer code $\mathcal{S}^G_c$.
In (b) and (c), we illustrate the construction for $c=1$ and $c=2$, respectively.
Qubits (black and red dots) are placed at vertices and on edges of the graph $G$.
We depict in red the support of the representative $X(v)$ of the logical Pauli operator $\overline L^G_c$ that is associated with the vertex $v\in V$.
}
\label{fig_graphcode}
\end{figure*}

\subsection{Computational complexity and graph theory}

NP is the complexity class of problems that can be solved in polynomial time using a non-deterministic Turing machine, which can perform multiple operations at the same time in parallel at every time step, and accepts if any one of the parallel operations leads to an accepting state.
A problem is NP-hard if every problem in NP can be reduced to it in polynomial time, and a problem is NP-complete if it is both NP-hard and in NP.

Let $G = (V,E)$ be a graph with the sets of vertices $V$ and edges $E$.
We say that a subset of vertices $V'\subseteq V$ is an independent set for the graph $G$ iff no two vertices in $V'$ are joined by an edge in $E$.
Moreover, we say that a collection $\mathcal A$ comprising subsets of the vertices of $G$ is an independent collection for $G$ iff any two different $A, A' \in \mathcal A$ are disjoint and no two vertices belonging to, respectively, $A$ and $A'$ are joined by an edge in $E$.
Note that an independent set is a special case of an independent collection.
We denote the size of a maximum independent set for $G$ by $\alpha(G)$ and refer to it as the independence number of $G$.
Then, the following decision problem is NP-complete~\cite{Garey1979}.  
\begin{quote}
\textbf{INDEPENDENT SET}\\
\noindent \textbf{Input}: A graph $G = (V, E)$ and a positive integer $a$.\\
\textbf{Question}: Is the independence number of $G$ greater or equal to $a$, i.e., $\alpha(G) \geq a$?
\end{quote}
Furthermore, for generic graphs and a positive real number $\epsilon$ the problem of approximating $\alpha(G)$ up to a multiplicative factor of $|V|^{1-\epsilon}$ is NP-hard~\cite{hastad1996clique, v003a006}.
We use this result in our reduction in the following section.

\section{Hardness of finding and approximating the $c$-disjointness}\label{sec:hardness}
\label{sec_hardness}

We start this section by constructing a CSS stabilizer code $\mathcal{S}^G_c$ and a logical Pauli operator $\overline{L}^G_c$ for any finite graph $G$ and positive integer $c$.
Then, we show that the $c$-disjointness $\dis[c]{\overline{L}^G_c}$ can be related to the independence number $\alpha(G)$ of the graph $G$.
This, in turn, allows us to prove our main theorem by reducing the problem of finding the maximum independent set for $G$ to the problem of finding the $c$-disjointness $\dis[c]{\overline{L}^G_c}$.

\subsection{Constructing the CSS stabilizer code}
\label{sec_code}

Let $G = (V,E)$ be a graph and $c$ be a positive integer constant.
To define a CSS stabilizer code $\mathcal{S}^G_c$ associated with the graph $G$, we first place $|V| \choose c-1$ qubits on every vertex $v\in V$ and on every edge $e\in E$; see Fig.~\ref{fig_graphcode} for an illustrative example.
We label each qubit by a pair $(v,\nu)$ or $(e,\nu)$, where $\nu \subseteq \left[|V|\right]$ is a subset of $c-1$ integers from $\left[|V|\right] = \{1,\ldots, |V|\}$.
Additionally, we place one or two qubits at each vertex $v$ depending on whether ${|V| \choose c-1}(\deg v +1)$ is odd or even respectively, where $\deg v$ denotes the degree of $v$.
We label those additional qubits by a pair $(v,i)$, where $i\in[2]$.
We denote by $P_q$ a Pauli operator $P = X,Z$ supported on qubit with label $q$.
We denote by $Q(v)$ and $Q(e)$ all the qubits placed at the vertex $v$ and on the edge $e$, i.e., all the qubits with labels $(v,*)$ and $(e,*)$, respectively.
We also denote by $Q$ the set of all the qubits, i.e.,
\begin{equation}
Q = \bigcup_{v\in V} Q(v) \cup \bigcup_{e\in E} Q(e).
\end{equation}

The stabilizer group $\mathcal{S}^G_c$ associated with $G$ and $c$ is
\begin{equation}
\mathcal{S}^{G}_c = \bigg\langle X(u) X(v), \prod_{q\in Q} Z_q \mathrel{\bigg |} u,v\in V \bigg\rangle,
\end{equation}
where $X(v)$ denotes a Pauli $X$ operator associated with a vertex $v\in V$, which we define as follows
\begin{equation}
X(v) = \prod_{q\in Q(v)} X_q \prod_{e\ni v} \bigg(\prod_{q \in Q(e)} X_{q}\bigg).
\end{equation}
Note that for any $v\in V$ the Pauli $X$-type operator $X(v)$ is supported on the even number of qubits.
Thus, $X(v)$ and $\prod_{q\in Q} Z_q$ commute and $\mathcal{S}^G_c$ is an Abelian subgroup of the Pauli group satisfying $-I\not\in\mathcal{S}^G_c$.
Moreover, for any $u,v\in V$ the Pauli $X$ operators $X(u)$ and $X(v)$ are the representatives of the same non-trivial logical Pauli operator, which we denote by $\overline L^G_c$.
The last statement follows from the fact that $X(u)$ and $X(v)$ commute with the all stabilizer operators, are not contained in the stabilizer group $\mathcal{S}^G_c$, and their product forms a stabilizer operator $X(u)X(v)$. 
For convenience, for any subset of vertices $A\subseteq V$ we define
\begin{equation}
X(A) = \prod_{v\in A} X(v).
\end{equation}
Note that if $|A| \equiv 0 \mod 2$, then $X(A)$ is a stabilizer operator; otherwise, $X(A)$ is a representative of the logical Pauli operator $\overline{L}^G_c$.

We remark that if the graph $G$ has at least $c+2$ vertices, i.e., $|V|\geq c+2$, then the stabilizer code $\mathcal{S}^G_c$ associated with $G$ and $c$ is error-detecting, i.e., its min-distance satisfies $d_\downarrow > 1$.
This, in turn, implies that the disjointness of $\mathcal{S}^G_c$ is greater than one, i.e., $\dis {\mathcal{S}^G_c} > 1$; see Lemma~2(ii) in Ref.~\cite{Jochym-OConnor2018}.
To establish the claim that $d_\downarrow > 1$, it suffices to check that every single-qubit Pauli operator anticommutes with some stabilizer operator from $\mathcal{S}^G_c$.
Since the stabilizer operator $\prod_{q\in Q} Z_q$ is supported on every qubit, any single-qubit Pauli $X$ or $Y$ operator anticommutes with it.
For any vertex $v\in V$ a Pauli $Z$ operator on qubit $(v,i)$ or $(v,\nu)$ anticommutes with a stabilizer operator $X(u)X(v)$ for any vertex $u\in V\setminus \{ v\}$ or $u\in V\setminus (\{ v\} \cup \nu)$, respectively.
Lastly, for any edge $e\in E$ a Pauli $Z$ operator on qubit $(e,\nu)$ anticommutes with a stabilizer operator $X(u)X(v)$ for any vertex $u \in  V\setminus (\{w,w'\} \cup \nu)$, where $w,w'\in V$ are two vertices incident to $e$.
Note that we use $|V| \geq c + 2$ to guarantee that the set $V\setminus (\{w,w'\} \cup \nu)$ is non-empty.

\subsection{Relating the $c$-disjointness to the independence number}

Let $\mathcal S^G_c$ and $\overline L^G_c$ be a stabilizer code and a logical Pauli operator, which are associated with the graph $G = (V,E)$ and the positive integer $c$.
We abuse terminology and say that a collection $\mathcal A$ comprising subsets of the vertices of $G$ is $c$-disjoint iff $\{ X(A) \mathrel{|} A\in\mathcal A\}$ is a $c$-disjoint collection for the logical operator $\overline L^G_c$ of the stabilizer code $\mathcal S^G_c$.
We refer to $X(A)$ as a representative from $\mathcal A$.
We remark that $\mathcal A$ is allowed to be a multiset.
In what follows, we establish three technical lemmas.

\begin{lemma}
\label{lemma_disjoint}
Let $G = (V,E)$ be a graph, $V' \subseteq V$ be an independent set and $c$ be a positive integer.
Then, a collection $\mathcal A =  \{ \{v\} \mathrel{|} v \in V'\}$ is $c$-disjoint.
Moreover, if $|V'| (c+1) \equiv 1 \mod 2$, then $\mathcal A \cup \{V'\}$ is $c$-disjoint, too.
\end{lemma}

\begin{proof}
If $c=1$, then $|V'| (c+1) \equiv 0 \mod 2$ and one can straightforwardly show that $\mathcal A$ is $c$-disjoint.
In what follows we assume that $c\geq 2$.
Let $v\in V$ be a vertex, $e\in E$ be an edge incident to two vertices $w,w'\in V$ , $\nu \subset V$ be a subset of $c-1$ vertices and $i\in [2]$.
First note that any qubit $(v,i)$ supports one or zero representatives from $\mathcal A$, depending on whether $v$ belongs to $V'$ or not.
Let $u\in V'$ be a vertex and $X(u)$ be a corresponding representative of $\overline{L}^G_c$.
If $X(u)$ is supported on the qubit $(v,\nu)$, then $u \in (\nu\cup \{v\})\cap V'$.
Thus, the qubit $(v,\nu)$ supports at most $|\nu| + 1 = c$ representatives from $\mathcal A$.
Similarly, $X(u)$ is supported on qubit $(e,\nu)$, then $u \in (\nu\cup \{w,w'\})\cap V'$.
Since $V'$ is an independent set, at most one of $w$ and $w'$ can belong to $V'$.
We then conclude that qubit $(e,\nu)$ supports at most $|\nu|+1 = c$ representatives from $\mathcal A$.
Thus, $\mathcal A$ is $c$-disjont.

Now, assume that $|V'| (c+1) \equiv 1 \mod 2$.
Then, $X(V')$ is a representative of $\overline{L}^G_c$.
It remains to show that no qubit supports more than $c$ representatives from $\mathcal A \cup \{V'\}$.
Since $\mathcal A$ is $c$-disjoint, thus a potential problem may only arise for qubits which already support $c$ representatives from $\mathcal A$.
This, however, cannot happen for any qubit $(v,i)$, as $c\geq 2$.
If the qubit $(v,\nu)$ supports $c$ representatives from $\mathcal A$, then $\nu \cup \{v\} \subseteq V'$, and subsequently the qubit $(v,\nu)$ does not support $X(V')$, as $|\nu \cup \{v\}|\equiv 0 \mod 2$.
Similarly, if the qubit $(e,\nu)$ supports $c$ representatives from $\mathcal A$, then it does not support $X(V')$.
We conclude that $\mathcal A \cup \{V'\}$ is $c$-disjoint.
\end{proof}

\begin{lemma}
\label{lemma_independent}
Let $G = (V,E)$ be a graph and $c$ be a positive integer.
Let $\mathcal A$ be a collection of subsets of the vertices of $G$ that is $c$-disjoint and satisfies $|\mathcal A|\geq 9c^3/2$.
Then, there exists an independent collection $\mathcal A' \subseteq \mathcal A$, such that
$|\mathcal A'| \geq |\mathcal A| - 3c(c-1)$.
\end{lemma}

\begin{proof}
If $c=1$, then one can straightforwardly show that $\mathcal A$ is an independent collection and, subsequently, one can set $\mathcal A' = \mathcal A$.
In what follows we assume that $c\geq 2$ and provide an explicit construction of $\mathcal A'$.

We start by finding the subsets of vertices $L_1,\ldots,L_{3(c-1)} \in \mathcal A$ and vertices $v_1,\ldots,v_{3(c-1)}$, such that $v_i \in L_i$ and $v_i \not\in L_j$ for any two different $i,j\in [3(c-1)]$.
Let $\mathcal A_1 = \mathcal A$.
For $i = 1, \ldots, 3(c-1)$ one proceeds inductively to find $L_i$ and $v_i$.
Namely,
\begin{enumerate}
\item choose $L_i$ to be the smallest subset of vertices in $\mathcal A_i$, i.e.,
$L_i = \argmin_{L \in \mathcal A_i} |L|$,
\item choose $v_i$ to be any vertex in $L_i\setminus\bigcup_{j=1}^{i-1} L_i$,
\item define $\mathcal A'_i$ to be the collection of the elements of $\mathcal A_i$ that do not contain $v_i$, i.e., $\mathcal A'_i = \{ L \in \mathcal A_i \mathrel{|} L \not\ni v_i\}$,
\item define $\mathcal A_{i+1}$ be the collection of the elements of $\mathcal A'_i$ that are not fully contained within $\bigcup_{j=1}^i L_j$, i.e.,
$\mathcal A_{i+1} = \{ L \in \mathcal A'_i \mathrel{|} L \setminus \bigcup_{j=1}^i L_j \neq \emptyset \}$.
\end{enumerate}
It suffices to show that $|\mathcal A_i| \geq 1$ for any $i\in[3(c-1)]$, as this would imply that $L_i$ and $v_i$ can be found in steps 1 and 2.

By construction, we have the following chain of inclusions
$\mathcal A_i \supseteq\mathcal A'_i \supseteq\mathcal A_{i+1}$.
We also have
\begin{equation}
\label{eq_chain_ineqs}
|\mathcal A_{i+1}| \geq |\mathcal A'_i| - ic \geq |\mathcal A_i| - (i+1)c.
\end{equation}
To show Eq.~\eqref{eq_chain_ineqs}, first note that for every $L\in\mathcal A$ containing the vertex $v_i$, the corresponding operator $X(L)$ is supported on qubit $(v_i, 1)$.
Since $\mathcal A$ is $c$-disjoint, there are at most $c$ elements of $\mathcal A$ that include $v_i$.
Thus, we obtain $|\mathcal A'_i| \geq |\mathcal A_i| - c$.
By construction, we have $|L_1|\leq \ldots\leq |L_i|$, which leads to
\begin{equation}
\bigg|\bigcup_{j=1}^{i} L_{j}\bigg| \leq \sum_{j=1}^{i} |L_j| \leq i|L_{i}|.
\end{equation}
This, in turn, allows us to upper bound the number of elements of $\mathcal A'_i$ that are contained in $\bigcup_{j=1}^i L_i$.
By definition of $L_i$ in step 1, for any $L\in \mathcal A'_i$ we have $|L| \geq |L_i|$.
Since $\mathcal A$ is $c$-disjoint, so is $\mathcal A'_i$ and we conclude that there are at most
\begin{equation}
c \bigg|\bigcup_{j=1}^{i} L_{j}\bigg|/|L_{i}| \leq ic
\end{equation}
elements of $\mathcal A'_i$ that are contained in $\bigcup_{j=1}^i L_i$.
This leads to $|\mathcal A_{i+1}| \geq |\mathcal A'_i| - ic$, which, in turn, establishes Eq.~\eqref{eq_chain_ineqs}.
By recursively using Eq.~\eqref{eq_chain_ineqs} we obtain
\begin{eqnarray}
|\mathcal A_i| &\geq& |\mathcal A_1| - (i(i+1)/2 - 1)c \\
&\geq& 9c^3/2 -(9 c^2 - 15 c + 4)c/2 > 1
\end{eqnarray}
for any $i\in[3(c-1)]$, where the last inequality holds for any positive integer $c$.

Now, let us consider 
\begin{equation}
\mathcal A' = \{ L\in\mathcal A \mathrel{|} \forall i\in [3(c-1)]: v_i \not\in L\}.
\end{equation}
Since $\mathcal A$ is $c$-disjoint, we obtain $|\mathcal A'| \geq |\mathcal A| - 3c(c-1)$.
To show that $\mathcal A'$ is an independent collection, we assume the contrapositive.
Thus, there exist $L,L' \in \mathcal A'$, such that
either (i) $L$ and $L'$ overlap and $v\in L \cap L'$ is the shared vertex,
or (ii) $L$ and $L'$ are disjoint, and $u\in L$ and $w\in L'$ are two vertices incident to the same edge $e\in E$.
Let us assume that the proposition (i) holds and
\begin{equation}
\mu = \{ i\in [3(c -1)] \mathrel{|} v\in L_i \}.
\end{equation}
Note that $|\mu| \leq c-2$; otherwise, qubit $(v,1)$ would support at least $c+1$ representatives corresponding to $L$, $L'$ and $L_i$ for $i\in \mu$.
Thus, we can find a subset of indices $\nu \subset [3(c-1)] \setminus \mu$ of cardinality $c-1$, such that $v\not\in L_i$ for $i\in \nu$.
This leads to a contradiction, as the qubit $(v,\nu)$ supports $c+1$ representatives corresponding to $L$, $L'$ and $L_i$ for $i\in \nu$.
Now, let us assume that the proposition (ii) holds and
\begin{equation}
\mu = \{ i\in [3(c - 1)] \mathrel{|} \{u,w\} \cap L_i \neq \emptyset \}.
\end{equation}
Note that $|\mu| \leq 2c-4$; otherwise, one of the qubits $(u,1)$ and $(w,1)$ would support at least $c+1$ representatives among the ones that correspond to $L$, $L'$ and $L_i$ for $i\in \mu$.
Thus, we can find a subset of $c-1$ indices $\nu \subset [3(c-1)] \setminus \mu$, such that $u,w\not\in L_i$  for $i\in \nu$.
This leads to a contradiction, as the qubit $(e,\nu)$ supports $c+1$ representatives corresponding to $L$, $L'$ and $L_i$ for $i\in \nu$.
\end{proof}

\begin{lemma}
\label{lemma_disjointness}
Let $G = (V,E)$ be a graph with no isolated vertices, such that $\alpha(G) \geq 9c^3/2$, where $c$ is a positive integer.
Then, the $c$-disjointness for the logical operator $\overline{L}^G_c$ of the stabilizer code $\mathcal{S}^G_c$ is given by
\begin{equation}
\label{eq_thm_disjointness}
\dis[c]{\overline{L}^G_c} = (\alpha (G)+b)/c,
\end{equation}
where $b = \alpha(G)(c+1) \mod 2$.
\end{lemma}

\begin{proof}
Let $\mathcal A$ be a $c$-disjoint collection for $\overline{L}^G_c$ of largest possible size, i.e.,
$|\mathcal A | = c \dis[c]{\overline{L}^G_c}$.
Lemma~\ref{lemma_disjoint} implies that a maximum independent set $V' \subseteq V$ for $G$ leads to a $c$-disjoint collection for $\overline{L}^G_c$ of size $\alpha(G) + b$, and thus, by definition of $\mathcal A$, we have $|\mathcal A| \geq  \alpha(G) + b$.
Let $\mathcal A'\subseteq \mathcal A$ be an independent collection within $\mathcal A$ of largest possible size.
Note that $|\mathcal A'| \leq \alpha(G)$, as the size of any independent collection cannot be greater than the independence number of $G$.
Thus, we arrive at
\begin{equation}
\label{eq_chain_ineq}
|\mathcal A| \geq  \alpha(G) + b \geq \alpha(G) \geq |\mathcal A'|.
\end{equation}
Note that $|\mathcal A| = |\mathcal A'|$ can hold only if $b=0$, and in that case we recover Eq.~\eqref{eq_thm_disjointness}.
In what follows, we thus consider the case $\mathcal A \setminus \mathcal A' \neq \emptyset$ and choose $L \in \mathcal A\setminus \mathcal A'$.

By definition, $\mathcal A'$ is the largest independent collection within $\mathcal A$.
Thus, Lemma~\ref{lemma_independent} implies that
\begin{equation}
|\mathcal A'| \geq |\mathcal A| - 3c(c-1) \geq 9c^3/2 - 3c(c-1)\geq c+1,
\end{equation}
where the last inequality holds for any positive integer $c$.
Subsequently, let $L_1,\ldots,L_{c+1}\in\mathcal A'$ be $c+1$ different subsets of vertices and select vertices $v_i \in L_i$ for $i\in [c+1]$.
We can rule out the possibility that $L \cap \bigcup\mathcal A' = \emptyset$, where for brevity we write $\bigcup\mathcal A' = \bigcup_{L'\in\mathcal A'} L'$.
Namely, assume $L \cap \bigcup\mathcal A' = \emptyset$.
In that case, there must exists an edge $e\in E$ incident to $u\in L$ and $v\in\bigcup\mathcal A'$;
otherwise, $\mathcal A' \cup \{ L\} \subseteq \mathcal A$ would be an independent collection larger than $\mathcal A'$, leading to a contradiction.
Without loss of generality, let $v_c = v$.
Then, qubit $(e, \{v_i\}_{i\in[c-1]})$ supports $c+1$ different logical operators $X(L)$, $X(L_1)$,\ldots, $X(L_c)$, which is in a contradiction with $\mathcal A$ being $c$-disjoint.
We thus conclude that $L \cap \bigcup\mathcal A' \neq \emptyset$.

We can show that $L \supseteq \bigcup\mathcal A'$ by assuming the contrapositive.
Since $L \cap \bigcup\mathcal A' \neq \emptyset$, there exists a vertex $v\in\bigcup\mathcal A'\setminus L$.
Without loss of generality, let $v_c \in L \cap \bigcup\mathcal A'$ and $v_{c+1} = v$.
Then, one of the qubits $(v_c,\{v_i\}_{i\in[c-1]})$ or $(v_{c+1},\{v_i\}_{i\in[c-1]})$ is guaranteed to support $X(L)$, on top of other $c$ different operators, either $X(L_1),\ldots, X(L_c)$ or $X(L_1),\ldots,X(L_{c-1}),X(L_{c+1})$, respectively.
This is in a contradiction with $\mathcal A$ being $c$-disjoint.
We thus conclude that $L \supseteq \bigcup\mathcal A'$.

Note that $c$ has to be even; otherwise, the qubit $(v_c,\{v_i\}_{i\in[c-1]})$ would support $c+1$ different logical operators $X(L), X(L_1),\ldots, X(L_c)$, leading to a contradiction.
If $c$ is even, then $b = \alpha(G) \mod 2$ and $\alpha(G) + b$ is even.
We also obtain that $\mathcal A \setminus \mathcal A' = \{ L\}$;
otherwise, there would exist at least two subsets of vertices $L,K\in \mathcal A \setminus \mathcal A'$ satisfying $L,K\supseteq\bigcup\mathcal A'$ and the qubit $(v_{c-1},\{v_i\}_{i\in[c-1]})$ would support at least $c+1$ different operators $X(L), X(K), X(L_1),\ldots, X(L_{c-1})$, leading to a contradiction.
We thus conclude that $|\mathcal A| = |\mathcal A'| + 1$.

There are two cases to consider: (i) $L = \bigcup\mathcal A'$ and (ii) $L \setminus \bigcup\mathcal A' \neq\emptyset$.
Recall that for any subset of vertices $L'\subseteq V$ the Pauli operator $X(L')$ is a representative of $\overline{L}^G_c$ iff $|L'| \equiv 1 \mod 2$.
First, assume that the proposition (i) holds.
Then, $|\mathcal A'| \equiv 1 \mod 2$;
otherwise, 
\begin{equation}
|L| = |\bigcup \mathcal A'| = \sum_{L'\in\mathcal A'} |L'| \equiv |\mathcal A'| \equiv 0 \mod 2,
\end{equation}
leading to a contradiction with $X(L)$ being a representative of $\overline{L}^G_c$.
Using $|\mathcal A| = |\mathcal A'| + 1$, Eq.~\eqref{eq_chain_ineq} and the fact that $\alpha(G) + b$ is even, we finally recover Eq.~\eqref{eq_thm_disjointness}.
Now, assume that the proposition (ii) holds.
Then, $\mathcal A' \cup \{ L \setminus \bigcup\mathcal A' \}$ is an independent collection;
otherwise, there would exist an edge $e\in E$ incident to vertices $u\in L \setminus \bigcup\mathcal A'$ and $v\in\bigcup\mathcal A'$, we could set $v_c = v$ and the qubit $(e,\{v_i\}_{i\in [c-1]})$ would support $c+1$ different logical operators $X(L),X(L_1),\ldots,X(L_c)$, leading to a contradiction.
Since the size of any independent collection is at most $\alpha(G)$, thus $|\mathcal A'| + 1 \leq \alpha(G)$.
Using $|\mathcal A| = |\mathcal A'| + 1$ and Eq.~\eqref{eq_chain_ineq}, we conclude that $b=0$ and subsequently recover Eq.~\eqref{eq_thm_disjointness}. 
\end{proof}

\subsection{Putting things together}

Now, we are ready to establish the main result of our work, which asserts that it is NP-complete to calculate (or even approximate) the $c$-disjointness for any positive integer constant $c$.

\begin{theorem}
\label{thm_main}
For any positive integer constant $c$ the decision problem $c$-DISJOINTNESS is NP-complete.  Furthermore, for any $\epsilon > 0$ it is NP-hard to approximate the $c$-disjointness to within a multiplicative factor of $n^{(1 - \epsilon)/(c+1)}$, where $n$ is the number of physical qubits of the stabilizer code.  
\end{theorem}

\begin{proof}
First note that INDEPENDENT SET is NP-complete, even if we restrict our attention to the graphs with no isolated vertices and the independence number greater or equal to $9c^3/2$, which we require in our proof.  
We can show that $c$-DISJOINTNESS is NP-hard by reducing INDEPENDENT SET to it.

Let $G = (V,E)$ be a graph satisfying $\alpha(G) \geq 9c^3/2$.
We construct a new graph $G' =(V',E')$ as follows:
for every vertex $v\in V$ we introduce two vertices $v_1, v_2\in V'$, and for every edge $(u,v)\in E$ we introduce four edges $(u_1,v_1),(u_1,v_2),(u_2,v_1),(u_2,v_2)\in E'$.
One can show that
\begin{equation}
\alpha(G') = 2\alpha(G).
\end{equation}
Following Section~\ref{sec_code}, we construct the stabilizer code $\mathcal{S}^{G'}_c$ and the logical operator $\overline{L}^{G'}_c$.
Since $\alpha(G') \equiv 0 \mod 2$, then from Lemma~\ref{lemma_disjointness} we obtain
\begin{equation}
c\dis[c]{\overline{L}^{G'}_c} = \alpha(G') = 2\alpha(G)
\end{equation}
and, subsequently, $\alpha(G) \geq a$ iff $c\dis[c]{\overline{L}^{G'}_c} \geq 2a$.
We conclude that the output of the $c$-DISJOINTNESS problem for the instance $\mathcal{S}^{G'}_c$, $\overline{L}^{G'}_c$ and $2a$ is the output of INDEPENDENT SET for $G$ and $a$, which, in turn, establishes that $c$-DISJOINTNESS is NP-hard.

Now we turn to approximation algorithms.
Given the graph $G = (V,E)$, the number of physical qubits in the stabilizer code $\mathcal S^G_c$ is
\begin{equation}
\textstyle
n \leq |V| \left( {|V|\choose c-1} +2 \right) + |E| {|V|\choose c-1} < |V|^{c+1},
\end{equation}
where we use $|E| \leq {|V| \choose 2}$.
Depending on the parity of $\alpha(G)(c+1)$, the size of the largest $c$-disjoint collection for the logical operator $\overline{L}^G_c$ is equal to either $\alpha(G)$ or $\alpha(G)+1$.
Thus, for any $\epsilon>0$ the ability to approximate the $c$-disjointness $\dis[c]{\overline{L}_c^G}$ to within a multiplicative factor of $n^{(1 - \epsilon)/(c+1)}$ implies that we can approximate $\alpha(G)$ to within a multiplicative factor of
$n^{(1 - \epsilon)/(c+1)} < |V|^{1 - \epsilon}$.
Because approximating $\alpha(G)$ is NP-hard, thus approximating the $c$-disjointness $\dis[c]{\overline{L}_c^G}$ is also NP-hard.

Finally, we show that $c$-DISJOINTNESS is in NP.
Let us consider a polynomial time verifier that takes as its witness a collection $\mathcal{A}$, which supposedly comprises representatives of the logical operator $\overline{L}$ specified by its representative $L\in\overline L$.
The verifier needs to check the following three conditions:
(i) the size of $\mathcal A$ is greater than $ca$
(ii) every qubit supports at most $c$ operators from $\mathcal{A}$,
and (iii) every element of $\mathcal{A}$ is a representative of $\overline{L}$.
Conditions (i) and (ii) can be easily checked in time polynomial in $n$.
Also, condition (iii) can be stated equivalently as follows:
for every $P \in \mathcal A$ the operator $PL$ belongs to the stabilizer group $\mathcal S$ specifying the stabilizer code.
Let $M$ be the binary matrix of size $(n - k) \times 2n$ that describes $\mathcal S$.
In order to check the last condition, for every $P \in \mathcal{A}$ the verifier appends to $M$ a row corresponding to a string of $2n$ bits representing $PL$ and computes its rank over the field $\mathbb{F}_2$.
The rank is equal to $n-k$ iff $P$ is a representative of $\overline{L}$.
Since finding the rank of a binary matrix can be done in time polynomial in its size, thus the witness can be verified in time polynomial in $n$, and, subsequently, $c$-DISJOINTNESS is in NP.
\end{proof}

\section{Disjointness in Practice}
\label{sec_practice}

In this section we express the problem of calculating the disjointness as a linear program with exponentially many variables, which allows us to find the disjointness in exponential time.
We use this linear program to find the disjointness of the $[\![14,3,3]\!]$ stabilizer code, and ultimately rule out
the existence of a transversal logical non-Clifford gate.

\subsection{A linear program for disjointness}

Let $c$ be a positive integer and $\overline{L}$ be a non-trivial logical Pauli operator for an $[\![n,k,d]\!]$ stabilizer code.
We can formulate an optimization problem whose optimal value is the $c$-disjointness $\dis[c]{\overline{L}}$.
Namely, for every representative $L$ of $\overline L$ we introduce a variable $x_L$, which admits non-negative integer values and represents the number of times the representative $L$ appears in a collection $\mathcal A \subseteq \overline L$.
We want to maximize the size of $\mathcal A$ divided by $c$, i.e., $\sum_{L\in \overline L} x_L/c$, subject to $\mathcal A$ being $c$-disjoint.
Note that this constraint can be equivalently phrased as follows: for every qubit $i\in [n]$ we have
$\sum_{L\in\overline L: \supp L \ni i} x_L/c \leq 1$.
Thus, we conclude that $\dis[c]{\overline{L}}$ can be found as a solution to an integer linear program with $2^{n-k}$ variables.

We can relax this integer linear program by removing the integrality constraint of each variable and allowing $x'_L = x_L/c$ to be a non-negative real number for every $L\in\overline L$.
We subsequently arrive at the following linear program
\begin{align} 
\label{eq:lp1}
&\text{maximize} && \Delta = \sum_{L \in \overline L} x'_L,\\
\label{eq:lp2}
&\text{subject to} &&\forall  i \in [n]: \sum_{L\in\overline L: \supp L \ni i} x'_L \leq 1,\\
\label{eq:lp3}
& && \forall L \in \overline L: x'_L \geq 0.
\end{align}

We now state the following theorem.

\begin{theorem}
\label{theorem_finite_c}
For any $[\![n,k,d]\!]$ stabilizer code and non-trivial logical Pauli operator $\overline{L}\in \mathcal L$,
the optimal value $\Delta^*$ attained by the linear program in Eqs.~\eqref{eq:lp1}-\eqref{eq:lp3}
is equal to the supremum of the $c$-disjointness for $\overline L$, i.e.,
\begin{equation}
\label{eq_optimal_value}
\Delta^* = \sup_{c\in \mathbb Z_{+}} \dis[c]{\overline L}.
\end{equation}
Moreover, the supremum of the $c$-disjointness for $\overline L$ is attained at some positive integer $c^* = 2^{\poly(n2^{n-k})}$, i.e.,
\begin{equation}
\label{eq_finite_supremum}
\sup_{c\in \mathbb Z_{+}} \dis[c]{\overline L} = \dis[c^*]{\overline L}.
\end{equation}
\end{theorem}

\begin{proof}
Let $c$ be any positive integer and $\mathcal A$ be a $c$-disjoint collection for $\overline L$ of the largest possible size, i.e., $|\mathcal A| = c \dis[c]{\overline L}$.
If $x_L$ is the number of times a representative $L\in\overline L$ appears in $\mathcal A$, then for every qubit $i\in [n]$ we have $\sum_{L\in\overline L: \supp L \ni i} x_L/c \leq 1$, as $\mathcal A$ is $c$-disjoint.
Moreover,
\begin{equation}
\Delta^* \geq \sum_{L\in\overline L} x_L/c = \dis[c]{\overline L},
\end{equation}
leading to
\begin{equation}
\label{eq_optimal1}
\Delta^* \geq \sup_{c\in\mathbb Z_+} \dis[c]{\overline L}.
\end{equation}

Note that there exists some finite optimal solution to the linear program in Eqs.~\eqref{eq:lp1}-\eqref{eq:lp3}, since every variable $x'_L$ appears in at least one constraint and, subsequently, every feasible solution is bounded, i.e., $x'_L \leq 1$ for all $L\in\overline L$.
Moreover, we can specify the linear program in Eqs.~\eqref{eq:lp1}-\eqref{eq:lp3} in matrix form 
$\max\{{\bf c^\T x \mathrel{|} x}\in\mathbb{R}^n, \bf Ax = b, x \geq 0\}$, where $\bf A$ is a binary matrix of size $2^{n - k} \times n$ that captures the support of every representative, and $\bf b$ and $\bf c$ are vectors of all ones of length $n$ and $2^{n-k}$, respectively.
The existence of a finite optimal solution and the fact that $\bf A$, $\bf b$ and $\bf c$ have rational entries implies that there exists a rational optimal solution $\{x_L^*\}_{L\in\overline L}$ whose bit size is polynomially bounded in terms of the bit sizes of $\bf A$, $\bf b$, and $\bf c$ \cite{dantzig1951maximization}.
Then, for every $L\in \overline L$ we can find positive integers $a_L^*$ and $b_L^*$, such that $x_L^* = a_L^* / b_L^*$ and $a_L^* b_L^* \leq 2^{\poly(n2^{n-k})}$.
Let $c^*$ be the least common multiple of all $b_L^*$s, i.e.,
\begin{equation}
c^* = \lcm(\{b^*_L\}_{L\in\overline L}).
\end{equation}
Clearly, $c^* < (2^{\poly(n2^{n-k})})^{2^{n-k}} = 2^{\poly(n2^{n-k})}$.
Let $\mathcal A^*$ be a collection of representatives of $\overline L$ obtained by taking $x^*_L c^*$ copies of each $L\in\overline L$.
Since $\mathcal A^*$ is $c^*$-disjoint and $|\mathcal A^*| = \sum_{L\in\overline L} x_L^*c^*$, we obtain that
\begin{equation}
\label{eq_optimal2}
\Delta^* = \sum_{L\in\overline L} x_L^* = |\mathcal A^*|/c^*
\leq \dis[c^*]{\overline L}
\leq \sup_{c\in\mathbb Z_+} \dis[c^*]{\overline L}.
\end{equation}
Thus, the inequalities in Eqs.~\eqref{eq_optimal1}~and~\eqref{eq_optimal2} have to become equalities, which implies Eqs.~\eqref{eq_optimal_value}-\eqref{eq_finite_supremum}.
\end{proof}

Theorem~\ref{theorem_finite_c} implies the following proposition.

\begin{proposition}
\label{proposition_disjointness_equiv}
The disjointness of any stabilizer code $\mathcal S$ satisfies the following equality
\begin{equation}
\label{eq_disjointness_equiv}
\dis {\mathcal S} \\=  \min_{\overline L \in \mathcal L} \sup_{c\in\mathbb Z_+} \dis[c]{\overline{L}}.
\end{equation}
\end{proposition}

\begin{proof}
By definition of the disjointness, we have $\dis{\mathcal S} = \sup_{c \in\mathbb Z_+}\min_{\overline L\in \mathcal L} \dis[c]{\overline L}$.
Theorem~\ref{theorem_finite_c} guarantees that for any non-trivial logical Pauli operator $\overline L\in \mathcal L$ the supremum of the $c$-disjointness for $\overline L$ is attained at some positive integer $c_{\overline L}$.
Let $c^*$ be the least common multiple of  all $c_{\overline L}$s, i.e.,
\begin{equation}
c^* = \lcm\left(\{c_L\}_{\overline L \in \mathcal L}\right).
\end{equation}
We then obtain
\begin{eqnarray}
\label{eq_dis_chain1}
\min_{\overline L \in \mathcal L} \dis[c^*]{\overline L}
&\leq& \sup_{c \in\mathbb Z_+}\min_{\overline L\in \mathcal L} \dis[c]{\overline L}
\leq \min_{\overline L\in \mathcal L}\sup_{c \in\mathbb Z_+} \dis[c]{\overline L}\\
\label{eq_dis_chain2}
&=& \min_{\overline L\in \mathcal L} \dis[c_{\overline L}]{\overline L}
\leq \min_{\overline L\in \mathcal L} \dis[c^*]{\overline L},
\end{eqnarray}
where we use the max-min inequality and the fact that $\dis[a]{\overline L} \leq \dis[ab]{\overline L}$ for any $\overline L \in \mathcal L$ and $a,b\in\mathbb Z_+$.
Thus, all the above inequalities in Eqs.~\eqref{eq_dis_chain1}-\eqref{eq_dis_chain2} have to become equalities, which, in turn, implies Eq.~\eqref{eq_disjointness_equiv}.
\end{proof}

We remark that the linear program in Eqs.~\eqref{eq:lp1}-\eqref{eq:lp3}, with $2^{n - k}$ variables and $n$ constraints, can be solved in time $O(2^{2.5 (n - k)})$.
Thus, the disjointness of the stabilizer code $\mathcal S$ can be found in time $O(2^{2.5 n - 0.5 k})$, as we can find it by solving this linear program for all $2^{2k} - 1$ logical Pauli operators for $\mathcal S$.

\subsection{An illustrative example}
\label{section:example}

\begin{figure}
\centering
\input{figures/14-qubit.tikz}
\caption{
The $[\![14,3,3]\!]$ stabilizer code is defined by placing qubits (white dots) on the vertices of a rhombic dodecahedron and introducing $X$-, $Y$- and $Z$-type stabilizer generators for every red, green and blue face, respectively.
We depict ``the azimuthal projection'', where the four corner qubits are identified.
Dots colored in red, green and blue correspond to Pauli $X$, $Y$ and $Z$ operators, respectively.
}
\label{fig:big_boy}
\end{figure}
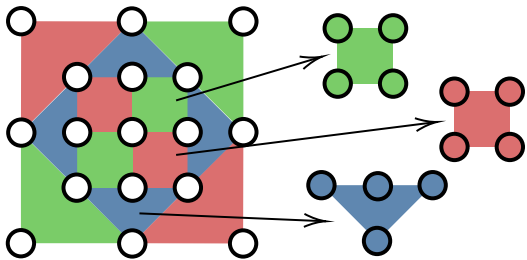

Now, we focus on an illustrative example of the recently introduced $[\![14,3,3]\!]$ stabilizer code~\cite{Landahl2020}; see Fig.~\ref{fig:big_boy}.
We would like to understand whether any non-Clifford logical operators for this code can be implemented via transversal gates.
If this was the case, then due to its small size the $[\![14,3,3]\!]$ stabilizer code could prove useful in, for instance, magic state distillation protocols.
Since the $[\![14,3,3]\!]$ stabilizer code is a non-CSS code, there are no off-the-shelf techniques to find transversal gates; however, we can use the disjointness to rule out the possibility of certain logical operations.

A brute-force approach to computing the $c$-disjointness, even for $c=1$ and small stabilizer codes, is infeasible.
In the case of the $[\![14,3,3]\!]$ stabilizer code, there are $4^3-1 = 63$ different non-trivial logical Pauli operators, and each of them has $2^{11} = 2048$ representatives.
Thus, for each logical Pauli operator there are $2^{2048} \approx 3.2\times 10^{616}$ possible subsets of its representatives, and we would need to check the qubit overlap for each of them (as the $1$-disjointness is achieved with a set rather than a multiset).

We can, however, use the linear program specified in Eqs.~\eqref{eq:lp1}-\eqref{eq:lp3}, which has $2048$ variables and $14$ constraints (excluding the positivity constraints).
We numerically find that the disjointness of the $[\![14,3,3]\!]$ stabilizer code is $2$; see the source code~\cite{Bostanci2020}.
We also find that the max-distance of the $[\![14,3,3]\!]$ stabilizer code is $d_\uparrow = 6$.
Thus, using the bound from Eq.~\eqref{eq_level_bound} we obtain that any transversal gate can only implement logical operations within the third level of the logical Clifford hierarchy.

\begin{figure}
\centering
\input{figures/14-qubit-paulis.tikz}
\caption{
For the $[\![14,3,3]\!]$ stabilizer code there are four logical Pauli operators with disjointness $2$.
We depict their smallest-weight representatives, where dots colored in red, green and blue correspond to Pauli $X$, $Y$ and $Z$ operators, respectively.
}
\label{fig:small_boys}
\end{figure}
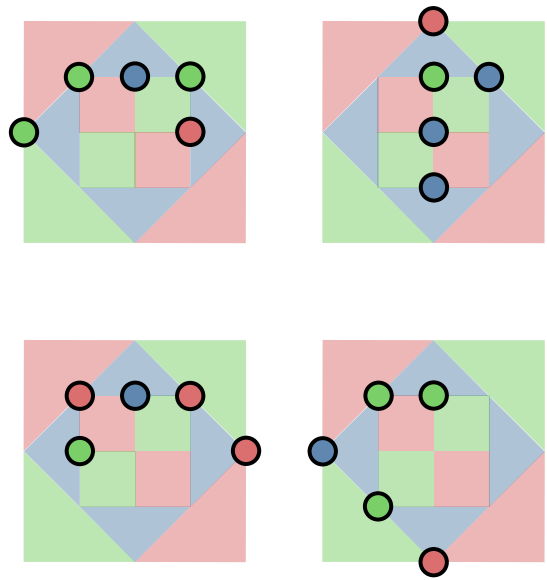

We now obtain an improvement of the bound in Eq.~\eqref{eq_level_bound}, which we subsequently use to rule out the possibility of any transversal gate that implements a non-Clifford logical operation for the $[\![14,3,3]\!]$ stabilizer code.  Recall that the main proof idea in Ref.~\cite{Jochym-OConnor2018} is to evaluate the (nested) group commutator of the transversal logical gate and $M$ logical Pauli operators for the stabilizer code $\mathcal S$.
If for every $M$-tuple of logical Pauli operators the resulting operator is a trivial logical operator, then the transversal logical gate is in the $M^{\text{th}}$ level of the logical Clifford hierarchy.
To recast this condition, it is useful to introduce the following quantity
\begin{equation}
\Omega_M(\mathcal S) = \max_{\{\overline L_i\}_{i\in [M]}} \min_{L_i \in \overline L_i}
\bigg|\bigcap_{i\in [M]} \supp L_i\bigg|,
\end{equation}
where for each $M$-tuple of logical Pauli operators $\{\overline L_i\}_{i\in [M]}$ we seek an $M$-tuple of corresponding representatives $L_i\in\overline L_i$, whose intersection $\bigcap_{i\in [M]} \supp L_i$ is the smallest.
Then, we can formulate the following strengthening of Theorem~5 from Ref.~\cite{Jochym-OConnor2018}.

\begin{theorem}
\label{thm:strengthening}
Consider a stabilizer code $\mathcal S$ with the min-distance $d_\downarrow$.
If $M$ is a positive integer satisfying
\begin{equation}
\Omega_M(\mathcal S) < d_\downarrow,
\end{equation}
then any transversal logical gate for $\mathcal S$ is in the $M^\text{th}$ level of the logical Clifford hierarchy.
\end{theorem}

To establish that for the $[\![14,3,3]\!]$ stabilizer code any transversal gate can only implement a logical Clifford gate, it suffices to show that $\Omega_2(\mathcal S) < 3$.
First, consider any pair of logical Pauli operators $\overline{L}_1$ and $\overline{L}_2$, for which, without loss of generality, $\dis{\overline{L}_1} > 2$.
Because the max-distance is $d_\uparrow = 6$, we can find a representative $L_2 \in \overline{L}_2$ of weight six.
By Lemma 4 from Ref.~\cite{Jochym-OConnor2018}, there exists a representative $L_1 \in \overline{L}_1$ such that $L_1$ and $L_2$ overlap on at most two qubits.
Thus, we restrict our attention to the logical Pauli operators $\overline{L}_1$ and $\overline{L}_2$, such that $\dis{\overline{L}_i} = 2$ for $i\in[2]$.
There are exactly four of these operators; see Fig.~\ref{fig:small_boys}.
By a brute-force search we verify that for all ten pairs of these operators there exists a choice of representatives that overlap on at most one qubit; see the source code~\cite{Bostanci2020}.
We finally conclude that $\Omega_2(\mathcal S) < 3$ for the $[\![14,3,3]\!]$ stabilizer code.

\section{Bounds on the disjointness}
\label{sec_bounds}

In this section, we provide bounds on the disjointness for the CSS codes, concatenated codes and hypergraph product codes, which, in turn, lead to the limitations on transversal gates available in each code family.

\subsection{CSS codes}

\begin{proposition}
\label{proposition_css}
Let $\mathcal S$ be a CSS code and $\overline L = \overline L^X \overline L^Z$ be a logical Pauli operator written in a standard logical basis.
Then, for any positive integer $c$ and $P\in \{X,Z\}$ the largest $c$-disjoint collection for $\overline L^P$ can be formed using $P$-type operators and the following inequalities hold
\begin{eqnarray}
\label{eq_css_ineq1}
\dis[c]{\overline L} &\leq& \min_{P\in \{X,Z\}}\dis[c]{\overline L^P},\\
\label{eq_css_ineq2}
\dis{\overline L} &\geq& \tfrac{\dis{\overline L^X}\dis{\overline L^Z}}{\dis{\overline L^X}+ \dis{\overline L^Z}-1},
\end{eqnarray}
where $\dis{\overline L}$, $\dis{\overline L^X}$ and $\dis{\overline L^Z}$ denote the supremum of the $c$-disjointness for $\overline L$, $\overline L^X$ and $\overline L^Z$, respectively.
\end{proposition}
\begin{proof}
Let $\mathcal A$ and $\mathcal A^P$ be the largest $c$-disjoint collections for $\overline L$ and $\overline L^P$, respectively.
Since $\overline L^P$ is a $P$-type logical Pauli operator, there exists a $P$-type operator $L^P$, which is a representative of $\overline L^P$.
Then, for every representative $L\in\mathcal A^P$ we can express it in the following way
\begin{equation}
L = L^P S^X_L S^Z_L,
\end{equation}
where $S^X_L$ and $S^Z_L$ are some $X$- and $Z$-type stabilizer generators of $\mathcal S$.
Note that $\{ L^P S^P_L \mathrel{|} L\in\mathcal A^P\}$ is an example of the desired $c$-disjoint collection for $\overline L^P$.
Similarly, for every representative $L\in\mathcal A$ we have
\begin{equation}
L = L^X S^X_L L^Z S^Z_L.
\end{equation}
Since $\{ L^P S^P_L \mathrel{|} L\in\mathcal A\}$ is a $c$-disjoint collection for $\overline L^P$, we conclude that $\dis[c] {\overline L} \leq \dis[c]{\overline L^P}$ and establish Eq.~\eqref{eq_css_ineq1}.

Theorem~\ref{theorem_finite_c} guarantees that we can find an integer $c^P$ and a $c^P$-disjoint collection $\mathcal A'^P$, such that $\left|\mathcal A'^P\right| = c^P \dis{\overline L^P}$.
Note that $\mathcal A' = \{ L L' \mathrel{|} L\in\mathcal A'^X, L'\in\mathcal A'^Z\}$ is a collection of representatives of $\overline L$ of size
\begin{equation}
|\mathcal A' | = c^X \dis{\overline L^X} c^Z \dis{\overline L^Z}.
\end{equation}
Since there are at most
\begin{eqnarray}
c' &=& c^X \left|\mathcal A'^Z\right| + c^Z \left|\mathcal A'^X\right|-c^Xc^Z\\
&=& c^X c^Z \left(\dis{\overline L^X} + \dis{\overline L^Z}-1\right)
\end{eqnarray}
elements of $\mathcal A'$ that are supported on any given qubit, thus $\mathcal A'$ is $c'$-disjoint.
This, in turn, allows us to conclude that $\dis{\overline L} \geq |\mathcal A'|/c'$ and establish Eq.~\eqref{eq_css_ineq2}.
\end{proof}

We remark that Proposition~\ref{proposition_disjointness_equiv} and Proposition~\ref{proposition_css} immediately imply that
\begin{equation}
\min_{\overline L \in \mathcal L^X \cup \mathcal L^Z} \dis{\overline L}/2
< \dis {\mathcal S}
\leq \min_{\overline L \in \mathcal L^X \cup \mathcal L^Z} \dis{\overline L},
\end{equation}
where $\mathcal L^P$ denotes the set of all non-trivial $P$-type logical operators in a standard logical basis for $P\in \{ X,Z \}$.
Thus, in order to obtain an approximation up to a multiplicative factor of two to the disjointness of any CSS code it suffices to consider the $c$-disjointness for its $X$- and $Z$-type logical operators.

\subsection{Concatenated stabilizer codes}

\begin{proposition}
\label{proposition_concatenated}
Let $\mathcal S_i$ be a stabilizer code with the min-distance $d^{(i)}_\downarrow > 1$, max-distance $d^{(i)}_\uparrow$ and disjointness $\dis {\mathcal S_i}$, where $i\in [2]$.
Then, the disjointness of the concatenated code $\mathcal S_1 \lhd \mathcal S_2$ satisfies the following inequality
\begin{equation}
\label{eq_disjointness_concat}
\dis{\mathcal S_1 \lhd \mathcal S_2} \geq \dis{\mathcal S_1} \dis{\mathcal S_2}.
\end{equation}
Moreover, if $M$ is the level of the logical Clifford hierarchy attainable by transversal logical gates for $\mathcal S_1 \lhd \mathcal S_2$, then
\begin{equation}
\label{eq_level_concat}
M\leq M_\mathrm{max} = \max_{i\in [2]} \left\lfloor\log_{\Delta^{(i)}} \left(d^{(i)}_\uparrow/d^{(i)}_\downarrow\right) \right\rfloor + 2.
\end{equation}
\end{proposition}

\begin{proof}
Let $\mathcal P = \{I, X,Y,Z\}$ be the set of single-qubit Pauli operators (modulo the phase) and $n_i$ be the number of physical qubits of the stabilizer code $\mathcal S_i$ for $i\in [2]$.
Let $K^P$ be a representative of the logical operator $\overline K^P$ for the stabilizer code $\mathcal S_2$, where we always choose $K^I$ to be the identity operator.
Note that if $L = \bigotimes_{i = 1}^{n_1} {P_i}$ implements a logical Pauli operator $\overline L$ for the stabilizer code $\mathcal S_1$, where $P_i \in \mathcal P$ for $i\in [n_1]$, then the following operator
\begin{equation}
L\lhd \left\{K^P\right\}_{P\in \mathcal P} = \bigotimes_{i=1}^{n_1} K^{P_i}
\end{equation}
implements $\overline L$ for the concatenated code $\mathcal S_1 \lhd \mathcal S_2$.

Using Theorem~\ref{theorem_finite_c} and Proposition~\ref{proposition_disjointness_equiv} we can find a positive integer $c^{(2)}$, such that for any $P\in\mathcal P\setminus \{I\}$ we have
\begin{equation}
\dis[c^{(2)}] {\overline K^P} \geq \dis {\mathcal S_2}.
\end{equation}
Subsequently, for $\overline K^P$ we can find a $c^{(2)}$-disjoint collection
\begin{equation}
\big\{ K^{P}_{i} \mathrel{\big |}  i\in [m] \big\},
\end{equation}
where $m = \lceil c^{(2)}\dis{\mathcal S_2}\rceil$.
Similarly, we can find a positive integer $c^{(1)}$, such that for any non-trivial logical Pauli operator $\overline L$ for the stabilizer code $\mathcal S_1$ we have
\begin{equation}
\dis[c^{(1)}]{\overline L} \geq \dis{\mathcal S_1}.
\end{equation}
This, in turn, implies the existence of a $c^{(1)}$-disjoint collection $\mathcal A$ for $\overline L$ that satisfies the following inequality
\begin{equation}
|\mathcal A|\geq \lceil c^{(1)}\dis {\mathcal S_1}\rceil.
\end{equation}

Let $\mathcal A'$ be a collection of representatives of $\overline L$ for the concatenated code $\mathcal S_1 \lhd \mathcal S_2$ defined as follows
\begin{equation}
\mathcal A' = \bigg\{ L \lhd \Big\{ K^P_i \Big\}_{P\in\mathcal P} 
\mathrel{\bigg |} L \in \mathcal A, i \in [m] \bigg\}.
\end{equation}
By construction, we have
\begin{equation}
|\mathcal A'| = m |\mathcal A| \geq c^{(1)} c^{(2)} \dis{\mathcal S_1} \dis{\mathcal S_2}.
\end{equation}
It is straightforward to show that $\mathcal A'$ is $c^{(1)}c^{(2)}$-disjoint.
Thus, we obtain
\begin{eqnarray}
\dis{\mathcal S_1\lhd \mathcal S_2} &\geq& \dis[c^{(1)} c^{(2)}]{\overline L}\\
&\geq& |\mathcal A'|/(c^{(1)} c^{(2)}) \geq \dis{\mathcal S_1} \dis{\mathcal S_2},
\end{eqnarray}
which establishes the inequality in Eq.~\eqref{eq_disjointness_concat}

Let $M_i = \log_{\dis{\mathcal S_i}} \left(d^{(i)}_\uparrow/d^{(i)}_\downarrow\right)$.
One can easily establish the following inequalities
\begin{eqnarray}
d_\uparrow(\mathcal S_1 \lhd \mathcal S_2) &\leq& d^{(1)}_\uparrow d^{(2)}_\uparrow,\\
d_\downarrow(\mathcal S_1 \lhd \mathcal S_2) &\geq& d^{(1)}_\downarrow d^{(2)}_\downarrow,
\end{eqnarray}
which, together with the inequality in Eq.~\eqref{eq_disjointness_concat}, lead to
\begin{eqnarray}
d_\uparrow(\mathcal S_1 \lhd \mathcal S_2) &\leq& d^{(1)}_\uparrow d^{(2)}_\uparrow
= d^{(1)}_\downarrow d^{(2)}_\downarrow \dis {\mathcal S_1}^{M_1} \dis {\mathcal S_2}^{M_2}\quad\\
&\leq& d_\downarrow(\mathcal S_1 \lhd \mathcal S_2) \dis {\mathcal S_1}^{M_1} \dis {\mathcal S_2}^{M_2}\\
&<& d_\downarrow(\mathcal S_1 \lhd \mathcal S_2)
\left(\dis {\mathcal S_1} \dis {\mathcal S_2}\right)^{M_\text{max}}\\
&\leq& d_\downarrow(\mathcal S_1 \lhd \mathcal S_2) \dis {\mathcal S_1 \lhd \mathcal S_2}^{M_\text{max}},
\end{eqnarray}
where we use the fact that $\dis{\mathcal S_i}> 1$ and $M_\text{max} > M_i$ for $i\in[2]$.
Then, using the inequality in Eq.~\eqref{eq_level_bound}, we obtain a bound on the level of the logical Clifford hierarchy attainable by transversal gates for $\mathcal S_1\lhd \mathcal S_2$, which is the inequality in Eq.~\eqref{eq_level_concat}.
\end{proof}

We remark that Proposition~\ref{proposition_concatenated} asserts that the level of the logical Clifford hierarchy attainable by transversal logical gates for the concatenated code $\mathcal S_1 \lhd \mathcal S_2$ cannot exceed the bounds in Eq.~\eqref{eq_level_bound} for the stabilizer codes $\mathcal S_1$ and $\mathcal S_2$.
However, these bounds are not necessarily saturated.
Thus, we cannot immediately rule out the possibility that transversal logical gates for $\mathcal S_1 \lhd \mathcal S_2$ may attain a level of the logical Clifford hierarchy higher than the levels attained by transversal logical gates for $\mathcal S_1$ and $\mathcal S_2$.
Lastly, we acknowledge that the bound in Eq.~\eqref{eq_disjointness_concat} has been previously independently derived in Ref.~\cite{Webster2020}.

\subsection{Hypergraph product codes}

Let $\mathcal S$ be a CSS code and $\mathcal L^P$ be the set of all non-trivial $P$-type logical operators in a standard logical basis, where $P\in \{ X,Z \}$.
We now introduce the $P$-type disjointness $\Delta^P(\mathcal{S})$ of the CSS code $\mathcal S$ as follows
\begin{equation}
\Delta^P(\mathcal{S}) = \sup_{c\geq 1} \min_{\overline{L} \in \mathcal{L}^P} \dis[c]{\overline{L}}.
\end{equation}
Similarly as for the disjointness, we can show that the $P$-type disjointness satisfies the following equality
\begin{equation}
\Delta^P(\mathcal{S}) = \min_{\overline{L} \in \mathcal{L}^Ps} \sup_{c\geq 1} \dis[c]{\overline{L}}.
\end{equation}
By definition, we have
\begin{equation}
\Delta^{P}(\mathcal S) \geq \dis{\mathcal S}.
\end{equation}

The $P$-type disjointness $\Delta^{P}(\mathcal S)$ is useful when the stabilizer group $\mathcal S$ is generated by only $P$-type operators.
In such a case, the min-distance of the stabilizer code $\mathcal S$ is one, i.e., $d_\downarrow = 1$.
Subsequently, Lemma~2(ii) in Ref.~\cite{Jochym-OConnor2018} implies that $\dis{\mathcal S} = 1$; however, $\Delta^{P}(\mathcal S)$ might still be greater than one.
This proves useful in the following proposition, where we establish an upper bound on the disjointness of hypergraph product codes.

\begin{proposition}
Let $H_i$ be a full-rank binary matrix of size $m_i\times n_i$ and $\mathcal S_i$ be a stabilizer code specified by the binary matrix $(H_i | 0_{m_i, n_i})$, where $i\in[2]$.
Let $\mathcal S$ be a hypergraph product code specified by the binary matrix in Eq.~\eqref{eq_hypergraph_matrix}. 
Then, the following inequality holds
\begin{equation}
\label{eq_disjoint_hyper}
\dis{\mathcal S} \leq \min_{i\in[2]} \Delta^X(\mathcal S_i).
\end{equation}
\end{proposition}

\begin{proof}
In what follows, we perform arithmetic operations modulo $2$.
Let $\spn M$ denote the row span of a binary matrix $M$.
Note that a Pauli operator $L_1$ is a representative of some non-trivial $X$-type logical operator $\overline L_1$ in a standard logical basis for the stabilizer code $\mathcal S_1$ iff $L_1$ is specified by the row vector $(l_1 | 0_{1,n_1})\in\{0,1\}^{2n_1}$, such that $l_1 \not\in\spn H_1$.
Let $l_2\in\{0,1\}^{m_2}$ be a non-zero row vector, such that $l_2 H_2 = 0_{1,n_2}$, and define
\begin{equation}
l = (l_1 \otimes l_2 | 0_{1,n_2 m_1}).
\end{equation}
Then, we have
\begin{equation}
\left(\one_{n_1}\otimes H_2^{\T} \middle| H_1^{\T} \otimes I_{n_2}\right) l^T = 0_{n_1 n_2,1}.
\end{equation}
Moreover, $l\not\in\spn \left(H_1 \otimes \one_{m_2} \middle | \one_{m_1} \otimes H_2\right)$; otherwise, we would obtain that $l_1 \in \spn H_1$, leading to a contradiction.
We thus conclude that a Pauli operator $L$ specified by the row vector $(l | 0_{1,n_1 m_2 + n_2 m_1})$ is a representative of some $X$-type logical operator for the stabilizer code $\mathcal S$.

Let $c$ be a positive integer and $\mathcal A$ be a $c$-disjoint collection for $\overline L$ of the largest possible size, i.e., $|\mathcal A| = c \dis[c]{\overline L}$.
Since the stabilizer code $\mathcal S$ is a CSS code and $\overline L$ is an $X$-type logical operator, thus Proposition~\ref{proposition_css} implies that we can select $\mathcal A$ in a way that it contains only $X$-type operators.
Let $\lambda\in [m_2]$ be a position of any non-zero entry of $l_2$ and define the index subset
\begin{equation}
\Lambda = \{ \lambda + (i-1)m_2 \mathrel{|} i\in[n_1] \}.
\end{equation}
We also treat $\Lambda$ as the subset of qubits whose indices belong to $\Lambda$.
By definition of the stabilizer group of the hypergraph product code, if $S$ is an $X$-type stabilizer operator from $\mathcal S$, then its restriction $S|_\Lambda$ to the subset of qubits $\Lambda$ is an $X$-type stabilizer operator from $\mathcal S_1$.
Moreover, for the representative $L$ of the logical operator $\overline L$, its restriction $L|_\Lambda$ is a representative of the logical operator $\overline L_1$.
Thus, a collection
\begin{equation}
\mathcal A_\Lambda = \{ K|_\Lambda \mathrel{|} K \in\mathcal A \}
\end{equation}
comprises representatives of $\overline L_1$.
Since $\mathcal A$ is $c$-disjoint, it has to be $c$-disjoint on the subset of qubits $\Lambda$, which immediately implies that $\mathcal A_\Lambda$ is $c$-disjoint.
Thus, we obtain
\begin{equation}
\dis[c]{\overline L_1} \geq |\mathcal A_\Lambda|/c = |\mathcal A|/c = \dis[c]{\overline L},
\end{equation}
and, subsequently, $\Delta^X(\mathcal S_1)\geq \dis{\mathcal S}$.

To establish the inequality in Eq.~\eqref{eq_disjoint_hyper}, we show that $\Delta^X(\mathcal S_2)\geq \dis{\mathcal S}$ in a similar way.
First, we note that for any $X$-type logical operator $\overline L_2$ for the stabilizer code $\mathcal S_2$ a Pauli operator $L_2$ is a representative of $\overline L_2$ iff $L_2$ is specified by a row vector $(l'_2|0_{1,n_2})\in \{0,1\}^{2n_2}$, such that $l'_2\not\in\spn H_2$.
Then, we choose $l'_1\in\{0,1\}^{m_1}$ to be any non-zero row vector, such that $l'_1H_1=0_{1,n_1}$, and define
$l' = (0_{1,n_1 m_2} | l'_1\otimes l'_2)$.
Finally, let $c$ be a positive integer.
We can then show that any $c$-disjoint collection for the logical operator $\overline L'$, whose representative is specified by the row vector $(l' | 0_{1,n_1 m_2 + n_2 m_1})$, gives rise to a $c$-disjoint collection for $\overline L_2$.
This, in turn, implies that $\dis[c]{\overline L_2} \geq \dis[c]{\overline L'}$, and, subsequently, 
$\Delta^X(\mathcal S_2)\geq \dis{\mathcal S}$.
\end{proof}

We remark that unlike other bounds in this section, Eq.~\eqref{eq_disjoint_hyper} provides an upper bound on the disjointness.
As such, it does not lead to a bound on the level of the Clifford hierarchy attainable by transversal logical gates for the hypergraph product code.

\section{Discussion}

The main result of our work, which is Theorem~\ref{thm_main}, established that for any positive integer constant $c$ the problem of calculating the $c$-disjointness (or even approximating it up to within a multiplicative factor) is NP-complete.
Although we have not shown that calculating the disjointness of stabilizer codes is hard, Theorem~\ref{thm_main} suggests so.
In general, our results indicate that finding fault-tolerant logical gates for generic quantum error-correcting codes is a computationally challenging task.
Other results presented in our work included: (i) formulating a linear program to calculate the disjointness, (ii) strengthening the main result of Ref.~\cite{Jochym-OConnor2018}, and (iii) providing bounds on the disjointness for various stabilizer code families.
Lastly, we remark that the source code~\cite{Bostanci2020} can be used to numerically estimate the disjointness of small stabilizer codes and, subsequently, rule out the existence of certain transversal logical gates.

We hope that our work initiates and motivates a thorough search of methods of finding the disjointness, as well as fault-tolerant logical gates at large.
We expect that there exist efficient algorithms to calculate or approximate the disjointness for certain code families, such as topological quantum codes.
We emphasize that clever usage of the underlying code symmetries might further simplify this problem.
For example, for codes that are invariant under certain permutations of qubits, the number of variables in the linear program in Eqs.~\eqref{eq:lp1}-\eqref{eq:lp3} to find the disjointness can be reduced from $2^n$ to $2^n/n$.
Moreover, it would be interesting to make the connection between the disjointness and other code quantities, such as the price~\cite{Pastawski2017}.

\acknowledgements{
A.K. thanks David Gosset, Andrew Landahl, Pooya Ronagh and Jamie Sikora for helpful discussions.
A.K. acknowledges funding provided by the Simons Foundation through the ``It from Qubit'' Collaboration.
Research at Perimeter Institute is supported in part by the Government of Canada through the Department of Innovation, Science and Economic Development Canada and by the Province of Ontario through the Ministry of Colleges and Universities.
This work was completed prior to A.K. joining AWS Center for Quantum Computing.}

\bibliography{biblio_complexity}

\end{document}

%% file: figures/14-qubit.tikz
\tikzset{every picture/.style={line width=0.75pt}} 

\begin{tikzpicture}[x=0.75pt,y=0.75pt,yscale=-1,xscale=1]

\draw  [draw opacity=0][fill={rgb, 255:red, 219; green, 111; blue, 111 }  ,fill opacity=1 ] (74.69,129.56) -- (130.67,73.48) -- (130.53,129.7) -- cycle ;
\draw  [draw opacity=0][fill={rgb, 255:red, 219; green, 111; blue, 111 }  ,fill opacity=1 ] (46.41,45.8) -- (74.33,45.8) -- (74.33,73.73) -- (46.41,73.73) -- cycle ;
\draw  [draw opacity=0][fill={rgb, 255:red, 122; green, 204; blue, 104 }  ,fill opacity=1 ] (74.58,45.8) -- (102.5,45.8) -- (102.5,73.73) -- (74.58,73.73) -- cycle ;
\draw  [draw opacity=0][fill={rgb, 255:red, 122; green, 204; blue, 104 }  ,fill opacity=1 ] (18.48,73.72) -- (74.69,129.56) -- (18.48,129.56) -- cycle ;
\draw  [draw opacity=0][fill={rgb, 255:red, 95; green, 135; blue, 176 }  ,fill opacity=1 ] (74.63,129.68) -- (46.59,101.52) -- (102.68,101.52) -- cycle ;
\draw  [draw opacity=0][fill={rgb, 255:red, 122; green, 204; blue, 104 }  ,fill opacity=1 ] (130.67,73.48) -- (74.46,17.64) -- (130.67,17.64) -- cycle ;
\draw  [draw opacity=0][fill={rgb, 255:red, 219; green, 111; blue, 111 }  ,fill opacity=1 ] (74.46,17.64) -- (18.48,73.72) -- (18.61,17.51) -- cycle ;
\draw  [draw opacity=0][fill={rgb, 255:red, 219; green, 111; blue, 111 }  ,fill opacity=1 ] (74.64,73.49) -- (102.56,73.49) -- (102.56,101.41) -- (74.64,101.41) -- cycle ;
\draw  [draw opacity=0][fill={rgb, 255:red, 122; green, 204; blue, 104 }  ,fill opacity=1 ] (46.84,73.6) -- (74.76,73.6) -- (74.76,101.52) -- (46.84,101.52) -- cycle ;
\draw  [draw opacity=0][fill={rgb, 255:red, 95; green, 135; blue, 176 }  ,fill opacity=1 ] (18.68,73.73) -- (46.84,45.68) -- (46.84,101.78) -- cycle ;
\draw  [draw opacity=0][fill={rgb, 255:red, 95; green, 135; blue, 176 }  ,fill opacity=1 ] (74.46,17.64) -- (102.5,45.8) -- (46.41,45.8) -- cycle ;
\draw  [draw opacity=0][fill={rgb, 255:red, 95; green, 135; blue, 176 }  ,fill opacity=1 ] (130.66,73.85) -- (102.5,101.9) -- (102.5,45.8) -- cycle ;
\draw  [fill={rgb, 255:red, 255; green, 255; blue, 255 }  ,fill opacity=1 ][line width=1.5]  (67.81,17.64) .. controls (67.81,13.98) and (70.79,11) .. (74.46,11) .. controls (78.12,11) and (81.1,13.98) .. (81.1,17.64) .. controls (81.1,21.31) and (78.12,24.28) .. (74.46,24.28) .. controls (70.79,24.28) and (67.81,21.31) .. (67.81,17.64) -- cycle ;
\draw  [fill={rgb, 255:red, 255; green, 255; blue, 255 }  ,fill opacity=1 ][line width=1.5]  (124.03,17.64) .. controls (124.03,13.98) and (127,11) .. (130.67,11) .. controls (134.34,11) and (137.31,13.98) .. (137.31,17.64) .. controls (137.31,21.31) and (134.34,24.28) .. (130.67,24.28) .. controls (127,24.28) and (124.03,21.31) .. (124.03,17.64) -- cycle ;
\draw  [fill={rgb, 255:red, 255; green, 255; blue, 255 }  ,fill opacity=1 ][line width=1.5]  (95.86,45.8) .. controls (95.86,42.13) and (98.84,39.16) .. (102.5,39.16) .. controls (106.17,39.16) and (109.15,42.13) .. (109.15,45.8) .. controls (109.15,49.47) and (106.17,52.44) .. (102.5,52.44) .. controls (98.84,52.44) and (95.86,49.47) .. (95.86,45.8) -- cycle ;
\draw  [fill={rgb, 255:red, 255; green, 255; blue, 255 }  ,fill opacity=1 ][line width=1.5]  (124.03,73.48) .. controls (124.03,69.82) and (127,66.84) .. (130.67,66.84) .. controls (134.34,66.84) and (137.31,69.82) .. (137.31,73.48) .. controls (137.31,77.15) and (134.34,80.13) .. (130.67,80.13) .. controls (127,80.13) and (124.03,77.15) .. (124.03,73.48) -- cycle ;
\draw  [fill={rgb, 255:red, 255; green, 255; blue, 255 }  ,fill opacity=1 ][line width=1.5]  (68,45.56) .. controls (68,41.9) and (70.97,38.92) .. (74.64,38.92) .. controls (78.31,38.92) and (81.28,41.9) .. (81.28,45.56) .. controls (81.28,49.23) and (78.31,52.21) .. (74.64,52.21) .. controls (70.97,52.21) and (68,49.23) .. (68,45.56) -- cycle ;
\draw  [fill={rgb, 255:red, 255; green, 255; blue, 255 }  ,fill opacity=1 ][line width=1.5]  (40.2,45.68) .. controls (40.2,42.01) and (43.17,39.04) .. (46.84,39.04) .. controls (50.51,39.04) and (53.48,42.01) .. (53.48,45.68) .. controls (53.48,49.35) and (50.51,52.32) .. (46.84,52.32) .. controls (43.17,52.32) and (40.2,49.35) .. (40.2,45.68) -- cycle ;
\draw  [fill={rgb, 255:red, 255; green, 255; blue, 255 }  ,fill opacity=1 ][line width=1.5]  (11.97,17.51) .. controls (11.97,13.84) and (14.95,10.87) .. (18.61,10.87) .. controls (22.28,10.87) and (25.26,13.84) .. (25.26,17.51) .. controls (25.26,21.17) and (22.28,24.15) .. (18.61,24.15) .. controls (14.95,24.15) and (11.97,21.17) .. (11.97,17.51) -- cycle ;
\draw  [fill={rgb, 255:red, 255; green, 255; blue, 255 }  ,fill opacity=1 ][line width=1.5]  (68,73.49) .. controls (68,69.82) and (70.97,66.84) .. (74.64,66.84) .. controls (78.31,66.84) and (81.28,69.82) .. (81.28,73.49) .. controls (81.28,77.15) and (78.31,80.13) .. (74.64,80.13) .. controls (70.97,80.13) and (68,77.15) .. (68,73.49) -- cycle ;
\draw  [fill={rgb, 255:red, 255; green, 255; blue, 255 }  ,fill opacity=1 ][line width=1.5]  (95.92,73.49) .. controls (95.92,69.82) and (98.9,66.84) .. (102.56,66.84) .. controls (106.23,66.84) and (109.21,69.82) .. (109.21,73.49) .. controls (109.21,77.15) and (106.23,80.13) .. (102.56,80.13) .. controls (98.9,80.13) and (95.92,77.15) .. (95.92,73.49) -- cycle ;
\draw  [fill={rgb, 255:red, 255; green, 255; blue, 255 }  ,fill opacity=1 ][line width=1.5]  (40.2,73.6) .. controls (40.2,69.93) and (43.17,66.96) .. (46.84,66.96) .. controls (50.51,66.96) and (53.48,69.93) .. (53.48,73.6) .. controls (53.48,77.27) and (50.51,80.24) .. (46.84,80.24) .. controls (43.17,80.24) and (40.2,77.27) .. (40.2,73.6) -- cycle ;
\draw  [fill={rgb, 255:red, 255; green, 255; blue, 255 }  ,fill opacity=1 ][line width=1.5]  (12.04,73.73) .. controls (12.04,70.06) and (15.01,67.09) .. (18.68,67.09) .. controls (22.35,67.09) and (25.32,70.06) .. (25.32,73.73) .. controls (25.32,77.4) and (22.35,80.37) .. (18.68,80.37) .. controls (15.01,80.37) and (12.04,77.4) .. (12.04,73.73) -- cycle ;
\draw  [fill={rgb, 255:red, 255; green, 255; blue, 255 }  ,fill opacity=1 ][line width=1.5]  (68.12,101.52) .. controls (68.12,97.86) and (71.09,94.88) .. (74.76,94.88) .. controls (78.43,94.88) and (81.4,97.86) .. (81.4,101.52) .. controls (81.4,105.19) and (78.43,108.17) .. (74.76,108.17) .. controls (71.09,108.17) and (68.12,105.19) .. (68.12,101.52) -- cycle ;
\draw  [fill={rgb, 255:red, 255; green, 255; blue, 255 }  ,fill opacity=1 ][line width=1.5]  (96.04,101.52) .. controls (96.04,97.86) and (99.02,94.88) .. (102.68,94.88) .. controls (106.35,94.88) and (109.32,97.86) .. (109.32,101.52) .. controls (109.32,105.19) and (106.35,108.17) .. (102.68,108.17) .. controls (99.02,108.17) and (96.04,105.19) .. (96.04,101.52) -- cycle ;
\draw  [fill={rgb, 255:red, 255; green, 255; blue, 255 }  ,fill opacity=1 ][line width=1.5]  (39.94,101.52) .. controls (39.94,97.86) and (42.92,94.88) .. (46.59,94.88) .. controls (50.25,94.88) and (53.23,97.86) .. (53.23,101.52) .. controls (53.23,105.19) and (50.25,108.17) .. (46.59,108.17) .. controls (42.92,108.17) and (39.94,105.19) .. (39.94,101.52) -- cycle ;
\draw  [fill={rgb, 255:red, 255; green, 255; blue, 255 }  ,fill opacity=1 ][line width=1.5]  (67.99,129.68) .. controls (67.99,126.02) and (70.97,123.04) .. (74.63,123.04) .. controls (78.3,123.04) and (81.28,126.02) .. (81.28,129.68) .. controls (81.28,133.35) and (78.3,136.33) .. (74.63,136.33) .. controls (70.97,136.33) and (67.99,133.35) .. (67.99,129.68) -- cycle ;
\draw  [fill={rgb, 255:red, 255; green, 255; blue, 255 }  ,fill opacity=1 ][line width=1.5]  (123.89,129.7) .. controls (123.89,126.03) and (126.87,123.06) .. (130.53,123.06) .. controls (134.2,123.06) and (137.18,126.03) .. (137.18,129.7) .. controls (137.18,133.37) and (134.2,136.34) .. (130.53,136.34) .. controls (126.87,136.34) and (123.89,133.37) .. (123.89,129.7) -- cycle ;
\draw  [fill={rgb, 255:red, 255; green, 255; blue, 255 }  ,fill opacity=1 ][line width=1.5]  (11.84,129.56) .. controls (11.84,125.9) and (14.81,122.92) .. (18.48,122.92) .. controls (22.15,122.92) and (25.12,125.9) .. (25.12,129.56) .. controls (25.12,133.23) and (22.15,136.21) .. (18.48,136.21) .. controls (14.81,136.21) and (11.84,133.23) .. (11.84,129.56) -- cycle ;
\draw  [draw opacity=0][fill={rgb, 255:red, 219; green, 111; blue, 111 }  ,fill opacity=1 ] (237.14,53.06) -- (265.06,53.06) -- (265.06,80.99) -- (237.14,80.99) -- cycle ;
\draw  [draw opacity=0][fill={rgb, 255:red, 95; green, 135; blue, 176 }  ,fill opacity=1 ] (198.08,128.43) -- (170.04,100.27) -- (226.13,100.27) -- cycle ;
\draw  [draw opacity=0][fill={rgb, 255:red, 122; green, 204; blue, 104 }  ,fill opacity=1 ] (178.26,21.1) -- (206.18,21.1) -- (206.18,49.02) -- (178.26,49.02) -- cycle ;
\draw  [fill={rgb, 255:red, 122; green, 204; blue, 104 }  ,fill opacity=1 ][line width=1.5]  (171.62,21.1) .. controls (171.62,17.43) and (174.59,14.46) .. (178.26,14.46) .. controls (181.93,14.46) and (184.9,17.43) .. (184.9,21.1) .. controls (184.9,24.77) and (181.93,27.74) .. (178.26,27.74) .. controls (174.59,27.74) and (171.62,24.77) .. (171.62,21.1) -- cycle ;
\draw  [fill={rgb, 255:red, 122; green, 204; blue, 104 }  ,fill opacity=1 ][line width=1.5]  (199.54,21.1) .. controls (199.54,17.43) and (202.52,14.46) .. (206.18,14.46) .. controls (209.85,14.46) and (212.82,17.43) .. (212.82,21.1) .. controls (212.82,24.77) and (209.85,27.74) .. (206.18,27.74) .. controls (202.52,27.74) and (199.54,24.77) .. (199.54,21.1) -- cycle ;
\draw  [fill={rgb, 255:red, 122; green, 204; blue, 104 }  ,fill opacity=1 ][line width=1.5]  (199.54,49.02) .. controls (199.54,45.36) and (202.52,42.38) .. (206.18,42.38) .. controls (209.85,42.38) and (212.82,45.36) .. (212.82,49.02) .. controls (212.82,52.69) and (209.85,55.67) .. (206.18,55.67) .. controls (202.52,55.67) and (199.54,52.69) .. (199.54,49.02) -- cycle ;
\draw  [fill={rgb, 255:red, 122; green, 204; blue, 104 }  ,fill opacity=1 ][line width=1.5]  (171.62,49.02) .. controls (171.62,45.36) and (174.59,42.38) .. (178.26,42.38) .. controls (181.93,42.38) and (184.9,45.36) .. (184.9,49.02) .. controls (184.9,52.69) and (181.93,55.67) .. (178.26,55.67) .. controls (174.59,55.67) and (171.62,52.69) .. (171.62,49.02) -- cycle ;
\draw  [fill={rgb, 255:red, 95; green, 135; blue, 176 }  ,fill opacity=1 ][line width=1.5]  (191.86,100.8) .. controls (191.86,97.13) and (194.84,94.16) .. (198.5,94.16) .. controls (202.17,94.16) and (205.15,97.13) .. (205.15,100.8) .. controls (205.15,104.47) and (202.17,107.44) .. (198.5,107.44) .. controls (194.84,107.44) and (191.86,104.47) .. (191.86,100.8) -- cycle ;
\draw  [fill={rgb, 255:red, 95; green, 135; blue, 176 }  ,fill opacity=1 ][line width=1.5]  (219.49,100.27) .. controls (219.49,96.6) and (222.46,93.63) .. (226.13,93.63) .. controls (229.8,93.63) and (232.77,96.6) .. (232.77,100.27) .. controls (232.77,103.94) and (229.8,106.91) .. (226.13,106.91) .. controls (222.46,106.91) and (219.49,103.94) .. (219.49,100.27) -- cycle ;
\draw  [fill={rgb, 255:red, 95; green, 135; blue, 176 }  ,fill opacity=1 ][line width=1.5]  (191.44,128.43) .. controls (191.44,124.76) and (194.42,121.79) .. (198.08,121.79) .. controls (201.75,121.79) and (204.73,124.76) .. (204.73,128.43) .. controls (204.73,132.1) and (201.75,135.07) .. (198.08,135.07) .. controls (194.42,135.07) and (191.44,132.1) .. (191.44,128.43) -- cycle ;
\draw  [fill={rgb, 255:red, 95; green, 135; blue, 176 }  ,fill opacity=1 ][line width=1.5]  (163.39,100.27) .. controls (163.39,96.6) and (166.37,93.63) .. (170.04,93.63) .. controls (173.7,93.63) and (176.68,96.6) .. (176.68,100.27) .. controls (176.68,103.94) and (173.7,106.91) .. (170.04,106.91) .. controls (166.37,106.91) and (163.39,103.94) .. (163.39,100.27) -- cycle ;
\draw  [fill={rgb, 255:red, 219; green, 111; blue, 111 }  ,fill opacity=1 ][line width=1.5]  (258.42,80.99) .. controls (258.42,77.32) and (261.4,74.34) .. (265.06,74.34) .. controls (268.73,74.34) and (271.71,77.32) .. (271.71,80.99) .. controls (271.71,84.65) and (268.73,87.63) .. (265.06,87.63) .. controls (261.4,87.63) and (258.42,84.65) .. (258.42,80.99) -- cycle ;
\draw  [fill={rgb, 255:red, 219; green, 111; blue, 111 }  ,fill opacity=1 ][line width=1.5]  (230.5,80.99) .. controls (230.5,77.32) and (233.47,74.34) .. (237.14,74.34) .. controls (240.81,74.34) and (243.78,77.32) .. (243.78,80.99) .. controls (243.78,84.65) and (240.81,87.63) .. (237.14,87.63) .. controls (233.47,87.63) and (230.5,84.65) .. (230.5,80.99) -- cycle ;
\draw  [fill={rgb, 255:red, 219; green, 111; blue, 111 }  ,fill opacity=1 ][line width=1.5]  (258.42,53.06) .. controls (258.42,49.4) and (261.4,46.42) .. (265.06,46.42) .. controls (268.73,46.42) and (271.71,49.4) .. (271.71,53.06) .. controls (271.71,56.73) and (268.73,59.71) .. (265.06,59.71) .. controls (261.4,59.71) and (258.42,56.73) .. (258.42,53.06) -- cycle ;
\draw  [fill={rgb, 255:red, 219; green, 111; blue, 111 }  ,fill opacity=1 ][line width=1.5]  (230.5,53.06) .. controls (230.5,49.4) and (233.47,46.42) .. (237.14,46.42) .. controls (240.81,46.42) and (243.78,49.4) .. (243.78,53.06) .. controls (243.78,56.73) and (240.81,59.71) .. (237.14,59.71) .. controls (233.47,59.71) and (230.5,56.73) .. (230.5,53.06) -- cycle ;
\draw    (96.75,57.5) -- (168.53,35.93) ;
\draw [shift={(170.45,35.35)}, rotate = 523.27] [color={rgb, 255:red, 0; green, 0; blue, 0 }  ][line width=0.75]    (10.93,-3.29) .. controls (6.95,-1.4) and (3.31,-0.3) .. (0,0) .. controls (3.31,0.3) and (6.95,1.4) .. (10.93,3.29)   ;
\draw    (96.75,85) -- (226.47,68.6) ;
\draw [shift={(228.45,68.35)}, rotate = 532.79] [color={rgb, 255:red, 0; green, 0; blue, 0 }  ][line width=0.75]    (10.93,-3.29) .. controls (6.95,-1.4) and (3.31,-0.3) .. (0,0) .. controls (3.31,0.3) and (6.95,1.4) .. (10.93,3.29)   ;
\draw    (78.25,114.5) -- (170.25,118.41) ;
\draw [shift={(172.25,118.5)}, rotate = 182.44] [color={rgb, 255:red, 0; green, 0; blue, 0 }  ][line width=0.75]    (10.93,-3.29) .. controls (6.95,-1.4) and (3.31,-0.3) .. (0,0) .. controls (3.31,0.3) and (6.95,1.4) .. (10.93,3.29)   ;

\end{tikzpicture}

%% file: figures/14-qubit-paulis.tikz
\tikzset{every picture/.style={line width=0.75pt}} 

\begin{tikzpicture}[x=0.75pt,y=0.75pt,yscale=-1,xscale=1]

\draw  [draw opacity=0][fill={rgb, 255:red, 219; green, 111; blue, 111 }  ,fill opacity=0.5 ] (85.69,292.56) -- (141.67,236.48) -- (141.53,292.7) -- cycle ;
\draw  [draw opacity=0][fill={rgb, 255:red, 219; green, 111; blue, 111 }  ,fill opacity=0.5 ] (57.41,208.8) -- (85.33,208.8) -- (85.33,236.73) -- (57.41,236.73) -- cycle ;
\draw  [draw opacity=0][fill={rgb, 255:red, 122; green, 204; blue, 104 }  ,fill opacity=0.5 ] (85.76,208.68) -- (113.68,208.68) -- (113.68,236.6) -- (85.76,236.6) -- cycle ;
\draw  [draw opacity=0][fill={rgb, 255:red, 122; green, 207; blue, 104 }  ,fill opacity=0.5 ] (29.48,236.72) -- (85.69,292.56) -- (29.48,292.56) -- cycle ;
\draw  [draw opacity=0][fill={rgb, 255:red, 95; green, 135; blue, 176 }  ,fill opacity=0.5 ] (85.63,292.68) -- (57.59,264.52) -- (113.68,264.52) -- cycle ;
\draw  [draw opacity=0][fill={rgb, 255:red, 122; green, 207; blue, 104 }  ,fill opacity=0.5 ] (141.67,236.48) -- (85.46,180.64) -- (141.67,180.64) -- cycle ;
\draw  [draw opacity=0][fill={rgb, 255:red, 219; green, 111; blue, 111 }  ,fill opacity=0.5 ] (85.46,180.64) -- (29.48,236.72) -- (29.61,180.51) -- cycle ;
\draw  [draw opacity=0][fill={rgb, 255:red, 219; green, 111; blue, 111 }  ,fill opacity=0.5 ] (85.76,236.6) -- (113.68,236.6) -- (113.68,264.52) -- (85.76,264.52) -- cycle ;
\draw  [draw opacity=0][fill={rgb, 255:red, 122; green, 204; blue, 104 }  ,fill opacity=0.5 ] (57.84,236.6) -- (85.76,236.6) -- (85.76,264.52) -- (57.84,264.52) -- cycle ;
\draw  [draw opacity=0][fill={rgb, 255:red, 95; green, 135; blue, 176 }  ,fill opacity=0.5 ] (29.68,236.73) -- (57.84,208.68) -- (57.84,264.78) -- cycle ;
\draw  [draw opacity=0][fill={rgb, 255:red, 95; green, 135; blue, 176 }  ,fill opacity=0.5 ] (85.46,180.64) -- (113.5,208.8) -- (57.41,208.8) -- cycle ;
\draw  [draw opacity=0][fill={rgb, 255:red, 95; green, 135; blue, 176 }  ,fill opacity=0.5 ] (141.66,236.85) -- (113.5,264.9) -- (113.5,208.8) -- cycle ;
\draw  [draw opacity=0][fill={rgb, 255:red, 219; green, 111; blue, 111 }  ,fill opacity=0.5 ] (236.36,131.56) -- (292.34,75.48) -- (292.2,131.7) -- cycle ;
\draw  [draw opacity=0][fill={rgb, 255:red, 219; green, 111; blue, 111 }  ,fill opacity=0.5 ] (208.07,47.8) -- (236,47.8) -- (236,75.73) -- (208.07,75.73) -- cycle ;
\draw  [draw opacity=0][fill={rgb, 255:red, 122; green, 204; blue, 104 }  ,fill opacity=0.5 ] (236,47.8) -- (263.92,47.8) -- (263.92,75.73) -- (236,75.73) -- cycle ;
\draw  [draw opacity=0][fill={rgb, 255:red, 122; green, 207; blue, 104 }  ,fill opacity=0.5 ] (180.14,75.72) -- (236.36,131.56) -- (180.14,131.56) -- cycle ;
\draw  [draw opacity=0][fill={rgb, 255:red, 95; green, 135; blue, 176 }  ,fill opacity=0.5 ] (236.3,131.68) -- (208.25,103.52) -- (264.35,103.52) -- cycle ;
\draw  [draw opacity=0][fill={rgb, 255:red, 122; green, 207; blue, 104 }  ,fill opacity=0.5 ] (292.34,75.48) -- (236.12,19.64) -- (292.34,19.64) -- cycle ;
\draw  [draw opacity=0][fill={rgb, 255:red, 219; green, 111; blue, 111 }  ,fill opacity=0.5 ] (236.12,19.64) -- (180.14,75.72) -- (180.28,19.51) -- cycle ;
\draw  [draw opacity=0][fill={rgb, 255:red, 219; green, 111; blue, 111 }  ,fill opacity=0.5 ] (236,75.73) -- (263.92,75.73) -- (263.92,103.65) -- (236,103.65) -- cycle ;
\draw  [draw opacity=0][fill={rgb, 255:red, 122; green, 204; blue, 104 }  ,fill opacity=0.5 ] (208.07,75.73) -- (236,75.73) -- (236,103.65) -- (208.07,103.65) -- cycle ;
\draw  [draw opacity=0][fill={rgb, 255:red, 95; green, 135; blue, 176 }  ,fill opacity=0.5 ] (180.35,75.73) -- (208.51,47.68) -- (208.51,103.78) -- cycle ;
\draw  [draw opacity=0][fill={rgb, 255:red, 95; green, 135; blue, 176 }  ,fill opacity=0.5 ] (236.12,19.64) -- (264.17,47.8) -- (208.07,47.8) -- cycle ;
\draw  [draw opacity=0][fill={rgb, 255:red, 95; green, 135; blue, 176 }  ,fill opacity=0.5 ] (292.33,75.85) -- (264.17,103.9) -- (264.17,47.8) -- cycle ;
\draw  [draw opacity=0][fill={rgb, 255:red, 219; green, 111; blue, 111 }  ,fill opacity=0.5 ] (236.36,292.56) -- (292.34,236.48) -- (292.2,292.7) -- cycle ;
\draw  [draw opacity=0][fill={rgb, 255:red, 219; green, 111; blue, 111 }  ,fill opacity=0.5 ] (208.51,208.68) -- (236.43,208.68) -- (236.43,236.6) -- (208.51,236.6) -- cycle ;
\draw  [draw opacity=0][fill={rgb, 255:red, 122; green, 204; blue, 104 }  ,fill opacity=0.5 ] (236.43,208.68) -- (264.35,208.68) -- (264.35,236.6) -- (236.43,236.6) -- cycle ;
\draw  [draw opacity=0][fill={rgb, 255:red, 122; green, 207; blue, 104 }  ,fill opacity=0.5 ] (180.14,236.72) -- (236.36,292.56) -- (180.14,292.56) -- cycle ;
\draw  [draw opacity=0][fill={rgb, 255:red, 95; green, 135; blue, 176 }  ,fill opacity=0.5 ] (236.3,292.68) -- (208.25,264.52) -- (264.35,264.52) -- cycle ;
\draw  [draw opacity=0][fill={rgb, 255:red, 122; green, 207; blue, 104 }  ,fill opacity=0.5 ] (292.34,236.48) -- (236.12,180.64) -- (292.34,180.64) -- cycle ;
\draw  [draw opacity=0][fill={rgb, 255:red, 219; green, 111; blue, 111 }  ,fill opacity=0.5 ] (236.12,180.64) -- (180.14,236.72) -- (180.28,180.51) -- cycle ;
\draw  [draw opacity=0][fill={rgb, 255:red, 219; green, 111; blue, 111 }  ,fill opacity=0.5 ] (236.43,236.6) -- (264.35,236.6) -- (264.35,264.52) -- (236.43,264.52) -- cycle ;
\draw  [draw opacity=0][fill={rgb, 255:red, 122; green, 204; blue, 104 }  ,fill opacity=0.5 ] (208.51,236.6) -- (236.43,236.6) -- (236.43,264.52) -- (208.51,264.52) -- cycle ;
\draw  [draw opacity=0][fill={rgb, 255:red, 95; green, 135; blue, 176 }  ,fill opacity=0.5 ] (180.35,236.73) -- (208.51,208.68) -- (208.51,264.78) -- cycle ;
\draw  [draw opacity=0][fill={rgb, 255:red, 95; green, 135; blue, 176 }  ,fill opacity=0.5 ] (236.12,180.64) -- (264.17,208.8) -- (208.07,208.8) -- cycle ;
\draw  [draw opacity=0][fill={rgb, 255:red, 95; green, 135; blue, 176 }  ,fill opacity=0.5 ] (292.33,236.85) -- (264.17,264.9) -- (264.17,208.8) -- cycle ;
\draw  [draw opacity=0][fill={rgb, 255:red, 219; green, 111; blue, 111 }  ,fill opacity=0.5 ] (85.69,131.56) -- (141.67,75.48) -- (141.53,131.7) -- cycle ;
\draw  [draw opacity=0][fill={rgb, 255:red, 219; green, 111; blue, 111 }  ,fill opacity=0.5 ] (57.41,47.8) -- (85.33,47.8) -- (85.33,75.73) -- (57.41,75.73) -- cycle ;
\draw  [draw opacity=0][fill={rgb, 255:red, 122; green, 204; blue, 104 }  ,fill opacity=0.5 ] (85.58,47.8) -- (113.5,47.8) -- (113.5,75.73) -- (85.58,75.73) -- cycle ;
\draw  [draw opacity=0][fill={rgb, 255:red, 122; green, 207; blue, 104 }  ,fill opacity=0.5 ] (29.48,75.72) -- (85.69,131.56) -- (29.48,131.56) -- cycle ;
\draw  [draw opacity=0][fill={rgb, 255:red, 95; green, 135; blue, 176 }  ,fill opacity=0.5 ] (85.63,131.68) -- (57.59,103.52) -- (113.68,103.52) -- cycle ;
\draw  [draw opacity=0][fill={rgb, 255:red, 122; green, 207; blue, 104 }  ,fill opacity=0.5 ] (141.67,75.48) -- (85.46,19.64) -- (141.67,19.64) -- cycle ;
\draw  [draw opacity=0][fill={rgb, 255:red, 219; green, 111; blue, 111 }  ,fill opacity=0.5 ] (85.46,19.64) -- (29.48,75.72) -- (29.61,19.51) -- cycle ;
\draw  [draw opacity=0][fill={rgb, 255:red, 219; green, 111; blue, 111 }  ,fill opacity=0.5 ] (85.33,75.73) -- (113.25,75.73) -- (113.25,103.65) -- (85.33,103.65) -- cycle ;
\draw  [draw opacity=0][fill={rgb, 255:red, 122; green, 204; blue, 104 }  ,fill opacity=0.5 ] (57.84,75.85) -- (85.76,75.85) -- (85.76,103.78) -- (57.84,103.78) -- cycle ;
\draw  [draw opacity=0][fill={rgb, 255:red, 95; green, 135; blue, 176 }  ,fill opacity=0.5 ] (29.68,75.73) -- (57.84,47.68) -- (57.84,103.78) -- cycle ;
\draw  [draw opacity=0][fill={rgb, 255:red, 95; green, 135; blue, 176 }  ,fill opacity=0.5 ] (85.46,19.64) -- (113.5,47.8) -- (57.41,47.8) -- cycle ;
\draw  [draw opacity=0][fill={rgb, 255:red, 95; green, 135; blue, 176 }  ,fill opacity=0.5 ] (141.66,75.85) -- (113.5,103.9) -- (113.5,47.8) -- cycle ;
\draw  [fill={rgb, 255:red, 122; green, 207; blue, 104 }  ,fill opacity=1 ][line width=1.5]  (106.92,47.56) .. controls (106.92,43.9) and (109.9,40.92) .. (113.56,40.92) .. controls (117.23,40.92) and (120.21,43.9) .. (120.21,47.56) .. controls (120.21,51.23) and (117.23,54.21) .. (113.56,54.21) .. controls (109.9,54.21) and (106.92,51.23) .. (106.92,47.56) -- cycle ;
\draw  [fill={rgb, 255:red, 95; green, 135; blue, 176 }  ,fill opacity=1 ][line width=1.5]  (79,47.56) .. controls (79,43.9) and (81.97,40.92) .. (85.64,40.92) .. controls (89.31,40.92) and (92.28,43.9) .. (92.28,47.56) .. controls (92.28,51.23) and (89.31,54.21) .. (85.64,54.21) .. controls (81.97,54.21) and (79,51.23) .. (79,47.56) -- cycle ;
\draw  [fill={rgb, 255:red, 122; green, 207; blue, 104 }  ,fill opacity=1 ][line width=1.5]  (23.04,75.73) .. controls (23.04,72.06) and (26.01,69.09) .. (29.68,69.09) .. controls (33.35,69.09) and (36.32,72.06) .. (36.32,75.73) .. controls (36.32,79.4) and (33.35,82.37) .. (29.68,82.37) .. controls (26.01,82.37) and (23.04,79.4) .. (23.04,75.73) -- cycle ;
\draw  [fill={rgb, 255:red, 122; green, 207; blue, 104 }  ,fill opacity=1 ][line width=1.5]  (50.76,47.8) .. controls (50.76,44.13) and (53.74,41.16) .. (57.41,41.16) .. controls (61.07,41.16) and (64.05,44.13) .. (64.05,47.8) .. controls (64.05,51.47) and (61.07,54.44) .. (57.41,54.44) .. controls (53.74,54.44) and (50.76,51.47) .. (50.76,47.8) -- cycle ;
\draw  [fill={rgb, 255:red, 219; green, 111; blue, 111 }  ,fill opacity=1 ][line width=1.5]  (106.92,75.49) .. controls (106.92,71.82) and (109.9,68.84) .. (113.56,68.84) .. controls (117.23,68.84) and (120.21,71.82) .. (120.21,75.49) .. controls (120.21,79.15) and (117.23,82.13) .. (113.56,82.13) .. controls (109.9,82.13) and (106.92,79.15) .. (106.92,75.49) -- cycle ;
\draw  [fill={rgb, 255:red, 122; green, 207; blue, 104 }  ,fill opacity=1 ][line width=1.5]  (229.79,47.68) .. controls (229.79,44.01) and (232.76,41.04) .. (236.43,41.04) .. controls (240.1,41.04) and (243.07,44.01) .. (243.07,47.68) .. controls (243.07,51.35) and (240.1,54.32) .. (236.43,54.32) .. controls (232.76,54.32) and (229.79,51.35) .. (229.79,47.68) -- cycle ;
\draw  [fill={rgb, 255:red, 95; green, 135; blue, 176 }  ,fill opacity=1 ][line width=1.5]  (257.53,47.8) .. controls (257.53,44.13) and (260.5,41.16) .. (264.17,41.16) .. controls (267.84,41.16) and (270.81,44.13) .. (270.81,47.8) .. controls (270.81,51.47) and (267.84,54.44) .. (264.17,54.44) .. controls (260.5,54.44) and (257.53,51.47) .. (257.53,47.8) -- cycle ;
\draw  [fill={rgb, 255:red, 219; green, 111; blue, 111 }  ,fill opacity=1 ][line width=1.5]  (229.48,19.64) .. controls (229.48,15.98) and (232.45,13) .. (236.12,13) .. controls (239.79,13) and (242.76,15.98) .. (242.76,19.64) .. controls (242.76,23.31) and (239.79,26.28) .. (236.12,26.28) .. controls (232.45,26.28) and (229.48,23.31) .. (229.48,19.64) -- cycle ;
\draw  [fill={rgb, 255:red, 95; green, 135; blue, 176 }  ,fill opacity=1 ][line width=1.5]  (229.67,75.49) .. controls (229.67,71.82) and (232.64,68.84) .. (236.31,68.84) .. controls (239.98,68.84) and (242.95,71.82) .. (242.95,75.49) .. controls (242.95,79.15) and (239.98,82.13) .. (236.31,82.13) .. controls (232.64,82.13) and (229.67,79.15) .. (229.67,75.49) -- cycle ;
\draw  [fill={rgb, 255:red, 95; green, 135; blue, 176 }  ,fill opacity=1 ][line width=1.5]  (229.79,103.52) .. controls (229.79,99.86) and (232.76,96.88) .. (236.43,96.88) .. controls (240.1,96.88) and (243.07,99.86) .. (243.07,103.52) .. controls (243.07,107.19) and (240.1,110.17) .. (236.43,110.17) .. controls (232.76,110.17) and (229.79,107.19) .. (229.79,103.52) -- cycle ;
\draw  [fill={rgb, 255:red, 219; green, 111; blue, 111 }  ,fill opacity=1 ][line width=1.5]  (51.2,208.68) .. controls (51.2,205.01) and (54.17,202.04) .. (57.84,202.04) .. controls (61.51,202.04) and (64.48,205.01) .. (64.48,208.68) .. controls (64.48,212.35) and (61.51,215.32) .. (57.84,215.32) .. controls (54.17,215.32) and (51.2,212.35) .. (51.2,208.68) -- cycle ;
\draw  [fill={rgb, 255:red, 95; green, 135; blue, 176 }  ,fill opacity=1 ][line width=1.5]  (79.12,208.68) .. controls (79.12,205.01) and (82.09,202.04) .. (85.76,202.04) .. controls (89.43,202.04) and (92.4,205.01) .. (92.4,208.68) .. controls (92.4,212.35) and (89.43,215.32) .. (85.76,215.32) .. controls (82.09,215.32) and (79.12,212.35) .. (79.12,208.68) -- cycle ;
\draw  [fill={rgb, 255:red, 219; green, 111; blue, 111 }  ,fill opacity=1 ][line width=1.5]  (106.86,208.8) .. controls (106.86,205.13) and (109.84,202.16) .. (113.5,202.16) .. controls (117.17,202.16) and (120.15,205.13) .. (120.15,208.8) .. controls (120.15,212.47) and (117.17,215.44) .. (113.5,215.44) .. controls (109.84,215.44) and (106.86,212.47) .. (106.86,208.8) -- cycle ;
\draw  [fill={rgb, 255:red, 122; green, 207; blue, 104 }  ,fill opacity=1 ][line width=1.5]  (51.2,236.6) .. controls (51.2,232.93) and (54.17,229.96) .. (57.84,229.96) .. controls (61.51,229.96) and (64.48,232.93) .. (64.48,236.6) .. controls (64.48,240.27) and (61.51,243.24) .. (57.84,243.24) .. controls (54.17,243.24) and (51.2,240.27) .. (51.2,236.6) -- cycle ;
\draw  [fill={rgb, 255:red, 219; green, 111; blue, 111 }  ,fill opacity=1 ][line width=1.5]  (135.03,236.48) .. controls (135.03,232.82) and (138,229.84) .. (141.67,229.84) .. controls (145.34,229.84) and (148.31,232.82) .. (148.31,236.48) .. controls (148.31,240.15) and (145.34,243.13) .. (141.67,243.13) .. controls (138,243.13) and (135.03,240.15) .. (135.03,236.48) -- cycle ;
\draw  [fill={rgb, 255:red, 122; green, 207; blue, 104 }  ,fill opacity=1 ][line width=1.5]  (201.86,208.68) .. controls (201.86,205.01) and (204.84,202.04) .. (208.51,202.04) .. controls (212.17,202.04) and (215.15,205.01) .. (215.15,208.68) .. controls (215.15,212.35) and (212.17,215.32) .. (208.51,215.32) .. controls (204.84,215.32) and (201.86,212.35) .. (201.86,208.68) -- cycle ;
\draw  [fill={rgb, 255:red, 122; green, 207; blue, 104 }  ,fill opacity=1 ][line width=1.5]  (229.67,208.56) .. controls (229.67,204.9) and (232.64,201.92) .. (236.31,201.92) .. controls (239.98,201.92) and (242.95,204.9) .. (242.95,208.56) .. controls (242.95,212.23) and (239.98,215.21) .. (236.31,215.21) .. controls (232.64,215.21) and (229.67,212.23) .. (229.67,208.56) -- cycle ;
\draw  [fill={rgb, 255:red, 95; green, 135; blue, 176 }  ,fill opacity=1 ][line width=1.5]  (173.7,236.73) .. controls (173.7,233.06) and (176.68,230.09) .. (180.35,230.09) .. controls (184.01,230.09) and (186.99,233.06) .. (186.99,236.73) .. controls (186.99,240.4) and (184.01,243.37) .. (180.35,243.37) .. controls (176.68,243.37) and (173.7,240.4) .. (173.7,236.73) -- cycle ;
\draw  [fill={rgb, 255:red, 122; green, 207; blue, 104 }  ,fill opacity=1 ][line width=1.5]  (201.61,264.52) .. controls (201.61,260.86) and (204.58,257.88) .. (208.25,257.88) .. controls (211.92,257.88) and (214.89,260.86) .. (214.89,264.52) .. controls (214.89,268.19) and (211.92,271.17) .. (208.25,271.17) .. controls (204.58,271.17) and (201.61,268.19) .. (201.61,264.52) -- cycle ;
\draw  [fill={rgb, 255:red, 219; green, 111; blue, 111 }  ,fill opacity=1 ][line width=1.5]  (229.66,292.68) .. controls (229.66,289.02) and (232.63,286.04) .. (236.3,286.04) .. controls (239.97,286.04) and (242.94,289.02) .. (242.94,292.68) .. controls (242.94,296.35) and (239.97,299.33) .. (236.3,299.33) .. controls (232.63,299.33) and (229.66,296.35) .. (229.66,292.68) -- cycle ;

\end{tikzpicture}